%% file: jrnl1up-2col.tex
\documentclass[journal]{IEEEtran}
\usepackage{color}
\usepackage{relsize}
\usepackage{widetext}

\ifCLASSINFOpdf
\else
\fi
\usepackage[cmex10]{amsmath}
\usepackage{amssymb,bbm}
\newtheorem{theorem}{Theorem}

\newtheorem{lemma}{Lemma}

\newtheorem{remark}{Remark}
\newtheorem{claim}{Claim}
\usepackage{graphics}

\begin{document}
%
\title{Is Non-Unique Decoding Necessary?}
%
%
%

\author{Shirin~Saeedi~Bidokhti,~\IEEEmembership{ Member,~IEEE,}
        and Vinod~M.~Prabhakaran,~\IEEEmembership{Member,~IEEE}
\thanks{S. Saeedi Bidokhti was with the School of Computer and Communication Sciences, Ecole Polytechnique F\'ed\'eral de Lausanne. She is now with the Institute for Communications Engineering, Technische Universit\"{a}t M\"{u}nchen
 (e-mail: shirin.saeedi@tum.de).}
\thanks{V. Prabhakaran is with the School of Technology and Computer Science,
Tata Institute of Fundamental Research
(e-mail: vinodmp@tifr.res.in).}
\thanks{S. Saeedi Bidokhti was
partially supported by the Swiss National Science Foundation fellowship no. 146617. Vinod M. Prabhakaran was partially supported by
a Ramanujan Fellowship from the Department of Science and Technology,
Government of India.}
\thanks{The material in this paper
was presented in part at the 2012 IEEE International Symposium on Information
Theory, Boston, MA.}
\thanks{Copyright (c) 2013 IEEE. Personal use of this material is permitted.  However, permission to use this material for any other purposes must be obtained from the IEEE by sending a request to pubs-permissions@ieee.org.}
}

\maketitle

\begin{abstract}
In multi-terminal
communication systems, signals carrying messages meant for different
destinations are often observed together at any given destination
receiver. Han and Kobayashi (1981) proposed a receiving strategy which
performs a joint unique decoding of messages of interest along with a subset
of messages which are not of interest.  It is now well-known that this
provides an achievable region which is, in general, larger than if the
receiver treats all messages not of interest as noise.  Nair and El
Gamal (2009) and Chong, Motani, Garg, and El Gamal (2008)
independently proposed a generalization called indirect or non-unique
decoding where the receiver uses the codebook structure of the messages to uniquely decode only its messages of interest. 
Non-unique decoding has since been used in various scenarios.

{The main result in this paper is to provide an interpretation and a systematic proof technique for why non-unique decoding, in all known cases where it
has been employed, can be replaced by a particularly designed joint unique decoding strategy, without any penalty from a rate region viewpoint.}
\end{abstract}


%
\IEEEpeerreviewmaketitle


\section{Introduction}

Coding schemes for multi-terminal systems with many information sources and many
destinations try to exploit the broadcast and interference nature of the
communication media. A consequence of this is that in many schemes the signals
received at a destination carry information, not only about messages that are
expected to be decoded at the destination (\textit{messages of interest}), but
also about messages that are not of interest to that destination.  

{Standard methods in
(random) code design (at the encoder) are rate splitting, superposition coding and
Marton's coding \cite{ElGamalKim, Marton79}. On the other hand,
standard decoding techniques are successive decoding and joint 
decoding \cite{ElGamalKim,HanKobayashi81}. In
\cite{HanKobayashi81}, Han and Kobayashi proposed a receiving strategy
which performs a \textit{joint decoding} of messages of interest along with a
subset of messages which are not of interest. We will refer to this receiving strategy as joint \textit{unique} decoding (and to the decoders as joint \textit{unique} decoders) to emphasize the fact that it seeks a unique choice not only for the messages of interest, but also for the rest of the messages being jointly decoded. It is now well-known
that employing such a joint unique decoder in the code design provides an achievable region which is, in general, larger
than if the receiver decodes the messages of interest while treating all messages not of interest as
noise.  
Recently, Nair and El Gamal \cite{NairElGamal09} and Chong,
Motani, Garg, and El Gamal \cite{ChongMotaniGargElGamal08}
independently proposed a generalization called \textit{indirect or non-unique
decoding} where the decoder looks for the unique messages of interest while using the codebook structure of all the messages (including the ones not of interest). Unlike the joint unique decoder, such a decoder does not necessarily uniquely decode messages not of interest, though it might narrow them down to a smaller list. 
We refer to such a decoder as a {non-unique decoder}. With such a distinction, non-unique decoders perform at least as well as joint-unique decoders. 
Coding schemes which employ non-unique decoders have since played a role in
achievability schemes in different multi-terminal problems such as
\cite{NNC11,ZaidiPiantanidaShamai12,NairElGamal10,BandemerElGamal11,ChiaElGamal11}.
It is of interest, therefore, to see if they can achieve  higher reliable transmission rates compared to codes that employ joint unique decoders.

In \cite{NairElGamal09}, the idea of non-unique (indirect) decoding is studied in
the context of broadcast channels with degraded message sets.  Nair
and El Gamal consider a $3$-receiver general broadcast channel where a
source communicates a common message $M_0$ to three receivers $Y_1$,
$Y_2$, and $Y_3$ and a private message $M_1$ only to one of the
receivers, $Y_1$ (Fig.~\ref{jrnl1-fig}).  \input{channel-2col} They
characterize an inner-bound to the capacity region of this problem
using non-unique decoding and show its tightness for some special cases.
It turns out that the same inner-bound of \cite{NairElGamal09} can be
achieved using a joint unique decoding strategy at all receivers.  The
equivalence of the rate region achievable by non-unique decoding and
that of joint unique decoding was observed in \cite{NairElGamal09}, but it
was arrived at by comparing single letter expressions for the two rate
regions. A similar equivalence was also noticed
  in~\cite{ChongMotaniGargElGamal08}, again by comparing single-letter
  expressions. For noisy network coding \cite{NNC11}, it has been shown that the same rate region can be obtained using joint unique decoding and without the use of non-unique decoding \cite{WuXie10, WuXieAllerton10, KramerITW11, HouKramerISIT12}.
{{It was also observed in \cite{ZaidiPiantanidaShamai12} that non-unique decoding is not essential to achieve the capacity region of certain state-dependent multiple access channels and joint unique decoding suffices.}}

In this paper, we will provide a proof technique which systematically shows an equivalence between the rate region achievable through non-unique decoders and joint unique decoders in several examples. In particular, our line of argument is applicable to all known instances where non-unique decoding has been employed in the literature as we discuss in Section \ref{jrnl1-examples}. Our technique is based on designing a special auxiliary joint unique decoder which replaces the non-unique decoder and sheds some light on why this equivalence holds. However, we would like to note that analysis using non-unique decoding is often simpler and gives a more compact representation of the rate-region -- a fact observed in \cite{NairElGamal09,ChongMotaniGargElGamal08} -- which still makes it a valuable tool for analysis. 

Three remarks follow.
{
\begin{remark}
The reader might wonder if such an equivalence holds on the rate-regions of schemes employing joint unique decoders and non-unique decoders more generally. While our proof technique is systematic and general, it is coupled with the random nature of the codebook generation and the encoder design. Indeed, any decoding scheme is coupled with the encoding scheme and therefore asking for a more general equivalence (for any encoding scheme) seems to be a challenging problem (even to properly pose).
\end{remark}

\begin{remark}
Non-unique decoders are usually easier to work with (analytically), and they capture the correct error events (conceptually). One might wonder what the advantages of joint unique decoders are.  It is generally interesting to know if certain messages may be uniquely decoded at a receiver at no rate-cost. In principle, such messages may be exploited to improve the encoding schemes.  We refer the interested reader to \cite{SaeediPrabhakaranDiggavi13} where an application of using joint unique decoders is illustrated in designing a  block Markov encoding scheme for the broadcast channel with degraded messages. 
\end{remark}
\begin{remark}
In a related line of research, \cite{BandemerElGamalKim12} proves optimality of non-unique decoding for general discrete memoryless interference channels, when encoding is restricted to randomly generated codebooks, superposition coding, and time sharing. The result of this paper and the techniques we develop indicate that the same performance can  be achieved by employing joint unique decoding, and that joint unique decoding is also optimal in the sense discussed in  \cite{BandemerElGamalKim12}.
\end{remark}
}

In Section \ref{jrnl1-problemstatement}, we develop our proof technique in the context of \cite{NairElGamal09}.
While much of the discussion in this paper is
confined to this framework,  we show in Section \ref{jrnl1-examples}  that the technique applies more generally.

\input{problemstatementafterrevision}

\input{examplesrevision}


\section{Conclusion}
{{We examined the non-unique decoding strategy of \cite {NairElGamal09}
where messages of interest are decoded jointly with other messages 
even when the decoder is unable to
disambiguate uniquely some of the messages which are not of interest to it.
We showed that  in all known cases where it has been employed, non-unique decoding can be replaced by the classic joint unique decoding strategy without any penalty from a rate region viewpoint.
We believe that this technique
may be applicable more generally to show the equivalence
of rate regions achievable using random coding employing
non-unique decoders and joint unique decoders.}}

\appendices
\input{Appendix}

\section*{Acknowledgement}
The authors would like to thank Suhas Diggavi for helpful comments and discussions. The authors would also like to thank Aaron Wagner and the anonymous reviewers for their comments that helped improve the manuscript.



%




\ifCLASSOPTIONcaptionsoff
  \newpage
\fi



\bibliographystyle{IEEEtran}
\bibliography{jrnl1.bib}

%



%
\begin{IEEEbiographynophoto}{Shirin Saeedi Bidokhti}
received the B.Sc. degree in electrical engineering from the University of Tehran, Iran,  in 2005, and the M.Sc and Ph.D. degrees in Communication Systems from the Ecole Polytechnique F\'{e}d\'{e}rale de Lausanne (EPFL), Switzerland, in 2007 and 2012, respectively. In 2013, she was awarded a Swiss National Science Foundation Prospective Researcher Fellowship. 

Since 2012, she has been a  postdoctoral researcher at the Institute for Communication Engineering, Technische Universit\"{a}t M\"{u}nchen (TUM). Her research interests include information theory and coding, multi-user communications systems, and network coding.\end{IEEEbiographynophoto}

\begin{IEEEbiographynophoto}{Vinod M. Prabhakaran}
received his Ph.D. in 2007 from the EECS Department, University of California, Berkeley. He was a Postdoctoral Researcher at the Coordinated Science Laboratory, University of Illinois, Urbana-Champaign from 2008 to 2010 and at Ecole Polytechnique F\'ed\'erale de Lausanne, Switzerland in 2011. In Fall 2011, he joined the Tata Institute of Fundamental Research, Mumbai, where he currently holds the position of a Reader. His research interests are in information theory, wireless communication, cryptography, and signal processing.

He has received the Tong Leong Lim Pre-Doctoral Prize and the Demetri Angelakos Memorial Achievement Award from the EECS Department, University of California, Berkeley, and the Ramanujan Fellowship from the Department of Science and Technology, Government of India.
\end{IEEEbiographynophoto}






\end{document}

%% file: channel-2col.tex
\begin{figure}[h]
\setlength{\unitlength}{0.1in} 
\centering 
\begin{picture}(27,7) 
\put(8.75,0){\framebox(8.5,6){$p(y_1,y_2,y_3|x)$}}
\put(1.25,1.5){\framebox(5,3){Encoder}}
\put(19.75,0.25){\framebox(6,1.5){Decoder}}
\put(19.75,2.25){\framebox(6,1.5){Decoder}}
\put(19.75,4.25){\framebox(6,1.5){Decoder}}
\put(25.75,1){\vector(1,0){4}}
\put(25.75,3){\vector(1,0){4}}
 \put(25.75,5){\vector(1,0){4}}
\put(-2.75,3){\vector(1,0){4}}
\put(6.25,3){\vector(1,0){2.5}}
\put(17.25,1){\vector(1,0){2.5}}
\put(17.25,3){\vector(1,0){2.5}}
 \put(17.25,5){\vector(1,0){2.5}}
\put(6.75,4) {$X$}
\put(-3.75,4) {$M_0,M_1$}
\put(17.75,1.5) {$Y_3$} 
\put(17.75,3.5){$Y_2$}
\put(17.75,5.5) {$Y_1$}
\put(26.25,1.5) {$M_0$} 
\put(26.25,3.5){$M_0$}
\put(26.25,5.5) {$M_0,M_1$}
\end{picture}
\caption{The 3-receiver broadcast channel with two degraded message sets: message $M_0$ is destined to all receivers and message $M_1$ is destined to receiver $Y_1$.} 
\label{jrnl1-fig} 
\end{figure}
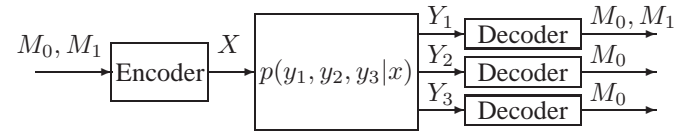

%% file: problemstatementafterrevision.tex
\section{Why Joint Unique Decoding Suffices in the Inner-Bound of Nair
and El Gamal in \cite{NairElGamal09}} \label{jrnl1-problemstatement}

We start this section by briefly reviewing the work of \cite{NairElGamal09}
where inner and outer bounds are derived for the capacity region of a
$3$-receiver broadcast channel with degraded message sets.  In particular,
we consider the case where a source communicates a common message (of rate $R_0$) to all
receivers, and a private message (of rate $R_1$) only to one of the receivers. 
A coding scheme is a sequence of $((2^{nR_0},2^{nR_1}),n)$ codes consisting of an encoder and a decoder
and is said to achieve a rate-tuple $(R_0,R_1)$ if the probability of error at the decoders decays to zero as $n$ grows large.
\paragraph*{Joint unique decoder vs. non-unique decoder}
We consider joint typical set decoding. A decoder at a certain destination may, in general, {\em examine} a subset of messages which includes, but is not necessarily limited to, the messages of interest to that destination. By the term examine, we mean that the decoder will try to make use of the structure (of the codebook) associated with the messages it examines. We say a coding scheme employs a {\em joint unique decoder} if the decoder tries to uniquely decode all the messages it considers (and declares an error if there is ambiguity in any of the messages, irrespective of whether such messages are of interest to the destination or not). In contrast, we say that a coding scheme employs a {\em non-unique decoder} if the decoder tries to decode uniquely only the messages of interest to the destination and tolerates ambiguity in messages which are not of interest.}

{{{Within this framework, Proposition
$5$ of \cite{NairElGamal09} establishes an achievable rate 
region for the problem of $3$-receiver broadcast channel with degraded message sets. The achievability is  through}}} a coding scheme that employs a non-unique decoder.
It turns out that employing a joint unique decoder, one can still
achieve the same inner-bound of \cite{NairElGamal09}.  
In this section, we develop a new proof technique to show this equivalence systematically. The same technique allows us to show the equivalence in all the examples considered in Section \ref{jrnl1-examples}. 


\subsection{Non-unique decoding in the achievable scheme of Nair  and El Gamal}
\label{jrnl1-indirectcodedesign}
The main problem studied in \cite{NairElGamal09} is that of sending two
messages over a $3$-receiver discrete memoryless broadcast channel
$p(y_1,y_2,y_3|x)$.  The source intends to communicate messages $M_0$ and
$M_1$ to receiver $1$ and message $M_0$ to receivers $2$ and $3$. Rates of
messages $M_0$ and $M_1$ are denoted by $R_0$ and $R_1$, respectively.
In \cite{NairElGamal09} an inner-bound to the capacity region is proved
using a standard encoding scheme based on superposition coding and Marton's
coding, and a non-unique decoding scheme called indirect decoding.
We briefly review this scheme.  

\subsubsection{Random codebook generation and encoding}  To design the
codebook, split the private message  $M_1$ into four independent parts $M_{10}$,
$M_{11}$, $M_{12}$, and $M_{13}$ of non-negative rates $S_0,S_1,S_2,S_3$,
respectively.  Let $R_1=S_0+S_1+S_2+S_3$, $T_2\geq S_2$ and $T_3\geq S_3$.
Fix a joint probability distribution $p(u,v_2,v_3,x)$.
Randomly and independently generate $2^{n(R_0+S_0)}$ sequences
$U^n(m_0,s_0)$, $m_0\in[1:2^{nR_0}]$ and $s_0\in[1:2^{nS_0}]$, each
distributed according to $\prod_ip_{U}(u_i)$.  For each
sequence $U^n(m_0,s_0)$, generate randomly and conditionally independently
($i$) $2^{nT_2}$ sequences $V_2^n(m_0,s_0,t_2)$, $t_2\in[1:2^{nT_2}]$, each
 according to $\prod_ip_{V_2|U}(v_{2i}|u_i)$, and ($ii$) $2^{nT_3}$ sequences $V_3^n(m_0,s_0,t_3)$,
$t_3\in[1:2^{nT_3}]$, each distributed according to $\prod_ip_{V_3|U}(v_{3i}|u_i)$. Randomly partition 
sequences $V^n_2(m_0,s_0,t_2)$ into $2^{nS_2}$ bins
$\mathcal{B}_2(m_0,s_0,s_2)$ and 
sequences
$V^n_3(m_0,s_0,t_3)$ into $2^{nS_3}$ bins $\mathcal{B}_3(m_0,s_0,s_3)$.  In
each product bin
$\mathcal{B}_2(m_0,s_0,s_2)\times\mathcal{B}_3(m_0,s_0,s_3)$, choose one
(random) jointly typical sequence pair
$(V^n_2(m_0,s_0,t_2),V^n_3(m_0,s_0,t_3))$. If there is no such pair,
declare an error whenever the message $(m_0,s_0,s_2,s_3)$ is to be transmitted.
Finally for each chosen jointly typical pair
$(V^n_2(m_0,s_0,t_2),V^n_3(m_0,s_0,t_3))$ in each product bin $(s_2,s_3)$,
randomly and conditionally independently generate $2^{nS_1}$ sequences
$X^n(m_0,s_0,s_2,s_3,s_1)$, $s_1\in[1:2^{nS_1}]$, each distributed
according to $\prod_ip_{X|UV_2V_3}(x_{i}|u_i,v_{2i},v_{3i})$.
To send the message pair $(m_0,m_1)$, where $m_1$ is expressed as
$(s_0,s_1,s_2,s_3)$, the encoder sends the codeword
$X^n(m_0,s_0,s_2,s_3,s_1)$.

\subsubsection {Non-unique decoding}
Receiver $Y_1$ jointly uniquely decodes all messages $M_0$, $M_{10}$, $M_{11}$,
$M_{12}$, and $M_{13}$. Receivers $Y_2$ and $Y_3$, however, decode $M_0$
indirectly, through a non-unique decoding scheme. More precisely,
\begin{itemize}
\item {}
Receiver $Y_1$ declares that the message tuple $(m_0,s_0,s_2,s_3,s_1)$ was
sent if it is the unique quintuple such that the received signal $Y^n_1$ is
jointly typical with $(U^n(m_0,s_0),V^n_2(m_0,s_0,t_2),V^n_3(m_0,s_0,t_3),$
$X^n(m_0,s_0,s_2,s_3,s_1))$, where  $s_2$ is the bin index of
$V_2^n(m_0,s_0,t_2)$ and  $s_3$ is the bin index of
$V_3^n(m_0,s_0,t_3)$.

\item{}
Receiver $Y_2$ declares that the message pair $(M_0,M_{10})=(m_0,s_0)$ was
sent if it finds a unique pair of indices $(m_0,s_0)$ for which the
received signal $Y^n_2$ is jointly typical with
$(U^n(m_0,s_0),V^n_2(m_0,s_0,t_2))$ for some index $t_2\in[1:2^{nT_2}]$. 
\item{}
Receiver $Y_3$ is similar to receiver $Y_2$ with $V_3$ and $t_3$,
respectively, instead of $V_2$ and $t_2$.
\end{itemize}

The above encoding/decoding scheme achieves rate pairs $(R_0,R_1)$ for
which inequalities \eqref{jrnl1-innerbound1} to \eqref{jrnl1-indirectY3}
below are satisfied for a joint distribution $p(u,v_2,v_3,x)$. The reader is referred to \cite{NairElGamal09} for the analysis of the
error probabilities.
\begin{eqnarray}
&&\hspace{-1.1cm}\text{Rate splitting constraints:}\nonumber\\
&& R_1=S_0+S_1+S_2+S_3\label{jrnl1-innerbound1}\\
&&T_2\geq S_2\\ 
&&T_3\geq S_3\\
&& S_0,S_1,S_2,S_3\geq 0 \\
&&\hspace{-1.1cm}\text{Encoding constraint:}\nonumber\\
&& T_2+T_3\geq S_2+S_3+I(V_2;V_3|U)\\
&&\hspace{-1.1cm}\text{Joint unique decoding constraints at receiver $Y_1$:}\nonumber\\
&&S_1\leq I(X;Y_1|U,V_2,V_3)\\
&&S_1+S_2\leq I(X;Y_1|UV_3)\\
&&S_1+S_3\leq I(X;Y_1|UV_2)\\
&&S_1+S_2+S_3\leq I(X;Y_1|U)\\
&&R_0+S_0+S_1+S_2+S_3\leq I(X;Y_1)\\
&&\hspace{-1.1cm}\text{Non-unique decoding constraint at receiver $Y_2$:}\nonumber\\
&&R_0+S_0+T_2\leq I(UV_2;Y_2)\label{jrnl1_indirectY2}\\
&&\hspace{-1.1cm}\text{Non-unique decoding constraint at receiver $Y_3$:}\nonumber\\
&&R_0+S_0+T_3\leq I(UV_3;Y_3).\label{jrnl1-indirectY3}
\end{eqnarray}
\subsection{Joint unique decoding suffices in the achievable scheme of Nair and El Gamal in \cite{NairElGamal09}}
\label{jrnl1_directdecoding}
Fix the codebook generation and encoding scheme to be that of Section
\ref{jrnl1-indirectcodedesign}. We will demonstrate {how a joint unique decoding scheme 
suffices by following these steps}:
\begin{itemize}
\item[(1)] We first analyze the non-unique decoder to characterize regimes
where it uniquely decodes all the messages it considers and regimes where it
decodes some of the messages non-uniquely.
\item[(2)] For each of the regimes, we deduce that the non-unique decoder may
be replaced by a joint unique decoder.
\end{itemize}
For the rest of this section, we only consider decoding schemes at receiver
$Y_2$. Similar arguments are valid for receiver $Y_3$ due to the symmetry
of the problem.  We refer to inequality \eqref{jrnl1_indirectY2}, which is shown in \cite{NairElGamal09} to ensure reliability of the non-unique decoder at receiver $Y_2$, as the non-unique decoding 
constraint \eqref{jrnl1_indirectY2}.

Let the rate pair $(R_0,R_1)$ be such that the non-unique decoder of receiver
$Y_2$ decodes message $M_0$ with high probability; i.e., the non-unique
decoding constraint \eqref{jrnl1_indirectY2} is satisfied. Consider the
following two regimes:
\begin{itemize}
\item [(a)]
 $R_0+S_0< I(U;Y_2)$,
\item[(b)]
 $R_0+S_0> I(U;Y_2)$.
\end{itemize}

In regime (a), it is clear from the defining condition that a joint unique decoder
which decodes $(M_0,M_{10})=(m_0,s_0)$ by finding the unique sequence $U^n(m_0,s_0)$
such that $(U^n(m_0,s_0),Y^n_2)$ is jointly typical will succeed  with high
probability. This is the joint unique decoder we may use in place of the non-unique
decoder for this regime. Notice that in this regime, while the non-unique decoder obtains
$(m_0,s_0)$ uniquely with high probability, it may not necessarily
succeed in uniquely decoding $t_2$. Indeed, in this regime insisting on
joint unique decoding of $U^n(m_0,s_0)$, $V^n_2(m_0,s_0,t_2)$ could,
in some cases, result in a strictly smaller achievable region.

Regime~(b) is the more interesting regime. Here it is clear that simply
decoding for $(M_0,M_{10})$ and treating all other messages as noise will not
work. Non-unique decoding must indeed be taking advantage of the codeword $V_2^n$ as
well. The non-unique decoder looks for a unique pair of messages $(m_0,s_0)$
such that there exists some $t_2$ for which
$(U^n(m_0,s_0),V^n(m_0,s_0,t_2),Y_2^n)$ is jointly typical. One may, in
general, expect that there could be several choices of $t_2$ even in this
regime. An important observation is that, in this regime, there is (with
high probability) only one choice for $t_2$. In other words, \emph{in this
regime, receiver 2 decodes $t_2$ uniquely along with $m_0$ and
$s_0$.} To see this, notice that using inequality \eqref{jrnl1_indirectY2} and (b) above, we have 
\begin{eqnarray}
T_2< I(V_2;Y_2|U).\label{jrnl1_indirectY2direct}
\end{eqnarray}
Inequalities \eqref{jrnl1_indirectY2} and \eqref{jrnl1_indirectY2direct}
together guarantee that a joint unique decoder can decode
messages $M_{0},M_{10}$, and $M_{12}$ with high probability. Note that condition \eqref{jrnl1_indirectY2} makes the  probability of an incorrect estimate for $(M_0,M_{10})$ vanish; and condition on $M_0,M_{10}$ being correctly estimated, inequality \eqref{jrnl1_indirectY2direct} drives the probability of an incorrect estimate for $M_{12}$ to zero. In other
words, in regime (b) the non-unique decoder ends up with a unique decoding of
the satellite codeword $V_2^n(m_0,s_0,t_2)$ with high probability; i.e., we
may replace the non-unique decoder with a joint unique decoder for messages $M_0$,
$M_{10}$, $M_{12}$. To summarize loosely, whenever the non-unique decoder is
forced to derive information from the codeword $V_2^n$ (i.e., when treating
$V_2^n$ as noise will not result in correct decoding), the non-unique decoder
will recover this codeword also uniquely. We make this loose intuition more
concrete in Section~\ref{onlinedecoder}.

The same argument goes through for receiver $Y_3$ and this shows that insisting on 
jointly uniquely decoding at all receivers is not restrictive in this problem. Thus,
we arrive at the following:
{\begin{theorem}
\label{jrnl1_theorem1}
For every rate pair $(R_0,R_1)$ satisfying the inner-bound of
\eqref{jrnl1-innerbound1}-\eqref{jrnl1-indirectY3}, there exists a coding scheme 
employing joint unique decoders which achieves the same rate pair.
\end{theorem}}

The idea behind the proof of Theorem~\ref{jrnl1_theorem1} was simple and
general. Consider a non-unique decoder which is decoding some messages of
interest. The message of interest in our problem is $M_0$.
Along with this message of interest, the decoder might also decode certain
other messages, $M_{10}$ and $M_{12}$ for example. The two main steps of
the proof is then as follows.
\begin{itemize}
 \item [($1$)] Analyze the non-unique decoder to determine what messages it
decodes uniquely. Depending on the regime of operation, the non-unique
decoder ends up uniquely decoding a subset of its intended messages, and
non-uniquely the rest of its intended messages. For example in regime (a)
above, the non-unique
decoder uniquely decodes only $M_0$ and $M_{10}$ and it might not be able
to settle on $M_{12}$. While in regime (b), the non-unique decoder ends up decoding all of its three messages
$M_0$, $M_{10}$, and $M_{12}$ uniquely.
\item [($2$)] In each regime of operation characterized in step ($1$), use
a joint unique decoder to only decode the messages that the non-unique decoder
uniquely decodes. In the above proof, this would be a joint unique decoder that
decodes $M_0$ and $M_{10}$ in regime (a) and a joint unique decoder
that decodes messages $M_0$, $M_{10}$, and $M_{12}$ in regime (b). Verify that the
resulting joint unique decoder does support the corresponding part of the 
rate region achieved by the non-unique decoding scheme. 
\end{itemize}

Though the idea is generalizable, analyzing the non-unique decoder in step
(1) is a tedious task.  Even for this very specific problem, it may not be
entirely clear how the condition dividing cases (a) and (b) can be derived.
Next, we try to resolve this using an approach which generalizes more
easily.

\subsection{An alternative proof to Theorem \ref{jrnl1_theorem1}: an auxiliary decoder}
\label{onlinedecoder}

We take an alternative approach in this section to prove Theorem \ref{jrnl1_theorem1}. 
The proof technique we present here has the same spirit as the proof in
Section~\ref{jrnl1_directdecoding}, but the task of determining which
subset of messages should be decoded in what regimes will be implicit
rather than explicit as before. To this end, we introduce an auxiliary
decoder which serves as a tool to help us develop the proof ideas. We do
not propose this more complicated auxiliary decoder as a new decoding
technique, but only as a proof technique to show sufficiency of joint unique 
decoding in the problem of \cite{NairElGamal09}. 
%
We analyze {the error probability of the} auxiliary decoder at receiver $Y_2$ and show that under the
random coding experiment, it decodes correctly with high probability if the
non-unique decoding constraint \eqref{jrnl1_indirectY2} holds. From this
auxiliary decoder and its performance, we will then be able to conclude
that there exists a joint unique decoding scheme that succeeds with high
probability.
  

We now define the auxiliary decoder. The auxiliary decoder at
receiver $Y_2$ is a more involved decoder which has access to two component
(joint unique) decoders:
\begin{itemize}
 \item  Decoder $1$ is a joint unique decoder which decodes messages $M_0$ and
$M_{10}$. It finds $M_0$, and $M_{10}$ by looking for
the unique sequence $U^n(m_0,s_0)$ for which the pair
$(U^n(m_0,s_0),Y^n_2)$ is jointly typical, and declares an error if there
exists no such unique sequence.
\item Decoder $2$ is a joint unique decoder which decodes messages $M_0$, $M_{10}$,
$M_{12}$. It finds $M_0$, $M_{10}$, $M_{12}$
by looking for the unique sequences $U^n(m_0,s_0)$ and $V_2^n(m_0,s_0,t_2)$
such that triple $(U^n(m_0,s_0),V_2^n(m_0,s_0,t_2),Y^n_2)$ is jointly
typical, and declares an error when such sequences do not exist.
\end{itemize}
The auxiliary decoder declares an error if either (a) both component decoders declare
errors, or (b) if both of them decode, but their decoded $(M_0,M_{10})$ messages
do not match. In all other cases it declares the $(M_0,M_{10})$ output of the
component decoder which did not declare an error as the decoded message. 

We analyze the error probability under the random coding experiment 
of such an auxiliary decoder at receiver $Y_2$ and prove that for any
$\epsilon>0$, there is a large enough $n$ such that
\begin{align}
 \Pr&(\text{error at the auxiliary decoder})\nonumber\\&\qquad\qquad\leq \epsilon+2^{n(R_0+S_0+T_2-I(UV_2;Y_2)+\gamma(\epsilon))}\!,\label{jrnl1-indirecterror}
\end{align}
where $\gamma(\epsilon)\to 0$ as $\epsilon\to 0$.
Inequality \eqref{jrnl1-indirecterror} shows that for large enough $n$ and
under the non-unique decoding constraint \eqref{jrnl1_indirectY2}, the
auxiliary decoder has an arbitrary small probability of failure.

{
{ 
To prove \eqref{jrnl1-indirecterror}, assume that $(m_0,s_0,s_1,s_2,s_3)=(1,1,1,1,1)$ is sent and indices
$t_1$ and $t_2$ in the encoding procedure are $(t_2,t_3)=(1,1)$. This assumption causes no loss of generality due to the symmetry of the codebook construction. We denote the random variables corresponding to these indices by $\mathcal{I}_{m_0}$, $\mathcal{I}_{s_0}$,$\ldots,$ $\mathcal{I}_{t_3}$ and we refer to the tuple $(\mathcal{I}_{m_0},\mathcal{I}_{s_0},\mathcal{I}_{s_1},\mathcal{I}_{s_2},\mathcal{I}_{s_3},\mathcal{I}_{t_2},\mathcal{I}_{t_3})$ by $\mathcal{I}$.
In the rest of this section, we assume $\mathcal{I}=\mathbf{1}$, the all $1$'s vector, and analyze the probability that receiver $Y_2$
declares $M_0\neq 1$. 
Receiver
$Y_2$ makes an error in decoding $M_0$ only if at least one of the
following events occur:
\begin{itemize}
 \item [$\mathcal{E}_1$:]\textit{The channel and/or the encoder is atypical:} the triple
$(U^n(1,1),V^n_2(1,1,1),V_3^n(1,1,1),$ $Y^n_2)$ is not jointly typical.
 \item [$\mathcal{E}_2$:]\textit{Both decoders fail to decode uniquely and declare
errors:} there are at least two distinct pairs $(\tilde{m}_0,\tilde{s}_0)$ and $(\breve{m}_0,\breve{s}_0)$
such that both pairs $(U^n(\tilde{m}_0,\tilde{s}_0),Y^n_2)$ and $(U^n(\breve{m}_0,\breve{s}_0),Y^n_2)$ are jointly typical; 
and similarly there are at least two distinct triples $(\hat{m}_0,\hat{s}_0,\hat{t}_2)$ and $(\check{m}_0,\check{s}_0,\check{t}_2)$ such that both 
triples $(U^n(\hat{m}_0,\hat{s}_0),V^n_2(\hat{m}_0,\hat{s}_0,\hat{t}_2),$ $Y^n_2)$ and $(U^n(\check{m}_0,\check{s}_0),V^n_2(\check{m}_0,\check{s}_0,\check{t}_2),$ $Y^n_2)$ are jointly typical.
\end{itemize}
Therefore, the probability that receiver $Y_2$ makes an error is upper-bounded in terms of the above events.
\begin{align}
\Pr&(\text{error at the auxiliary decoder}|\mathcal{I}=\mathbf{1})\nonumber\\
&\qquad\qquad\leq\Pr(\mathcal{E}_1|\mathcal{I}=\mathbf{1})+\Pr(\mathcal{E}_2\cap\bar{\mathcal{E}}_1|\mathcal{I}=\mathbf{1})\nonumber\\
& \qquad\qquad\leq\epsilon+\Pr(\mathcal{E}_2\cap\bar{\mathcal{E}}_1|\mathcal{I}=\mathbf{1}).\label{jrnl1-pe-bound}
\end{align}
where \eqref{jrnl1-pe-bound} follows because $\Pr(\mathcal{E}_1|\mathcal{I}=\mathbf{1})=\Pr((U^n(1,1),V^n_2(1,1,1),V_3^n(1,1,1),Y^n_2)\notin A^n_\epsilon|\mathcal{I}=\mathbf{1})\leq \epsilon$ (ensured by the encoding and 
the Asymptotic Equipartition 
Property). To upper-bound $\Pr(\mathcal{E}_2\cap\bar{\mathcal{E}}_1|\mathcal{I}=\mathbf{1})$, we write
\begin{align}
\nonumber\\
&\Pr(\mathcal{E}_2\cap\bar{\mathcal{E}}_1|\mathcal{I}=\mathbf{1})\\
&\stackrel{(a)}{\leq}\Pr\!\!\left(\hspace{-.3cm} \left.\begin{array}{c} ( U^n (1,1),V_2^n (1,1,1),V_3^n (1,1,1),Y^n_2\!)\!\in\! A^n_\epsilon ,\\\text{ and}\\(U^n(\tilde{m}_0,\tilde{s}_0),Y^n_2)\in A^n_\epsilon\\ \quad\quad  \text{for some }(\tilde{m}_0,\tilde{s}_0)\neq(1,1),\\\text{ and } \\(U^n(\hat{m}_0,\hat{s}_0),V^n_2(\hat{m}_0,\hat{s}_0,\hat{t}_2),Y^n_2)\!\in\! A^n_\epsilon\\ \quad\quad \text{for some }(\hat{m}_0,\hat{s}_0,\hat{t}_2)\neq(1,1,1)\end{array}\hspace{-.2cm}\right|\mathcal{I}\!=\!\mathbf{1} \hspace{-.15cm}\right)\nonumber\\\label{jrnl1-errorauxiliary}\\
&\leq\Pr\!\!\left(\hspace{-.3cm}\left.\begin{array}{c} ( U^n (1,1),V_2^n (1,1,1),V_3^n (1,1,1),Y^n_2\!)\!\in\! A^n_\epsilon ,\\\text{ and}\\ (U^n(\tilde{m}_0,\tilde{s}_0),Y^n_2)\in A^n_\epsilon\\ \quad\quad \text{for some }(\tilde{m}_0,\tilde{s}_0)\neq(1,1),\\\text{ and } \\(U^n(\hat{m}_0,\hat{s}_0),V^n_2(\hat{m}_0,\hat{s}_0,\hat{t}_2),Y^n_2)\!\in\! A^n_\epsilon\\ \quad\quad \text{for some }(\hat{m}_0,\hat{s}_0)\neq(1,1)\text{ and }\ \hat{t}_2\end{array} \hspace{-.2cm}\right|\mathcal{I}\!=\!\mathbf{1} \hspace{-.15cm}\right)
\nonumber\\
&\quad+\Pr\!\!\left(\hspace{-.3cm}\left.\begin{array}{c} ( U^n (1,1),V_2^n (1,1,1),V_3^n (1,1,1),Y^n_2\!)\!\in\! A^n_\epsilon ,\\\text{ and}\\
(U^n(\tilde{m}_0,\tilde{s}_0),Y^n_2)\!\in\! A^n_\epsilon\\ \quad\quad \text{for some
}(\tilde{m}_0,\tilde{s}_0)\neq(1,1),\\\text{ and }\\ \text{all }(U^n(\hat{m}_0,\hat{s}_0),V^n_2(\hat{m}_0,\hat{s}_0,\hat{t}_2),Y^n_2)\!\in\!
A^n_\epsilon\\\text{are s.t. } (\hat{m}_0,\hat{s}_0)\!=\!(1,1)\\\text{ with at
least one s.t.}\, \hat{t}_2\!\neq\! 1\end{array}\hspace{-.2cm}\right|\mathcal{I}\!=\!\mathbf{1} \hspace{-.15cm}\right)\nonumber\\
\label{jrnl1-onlinedecode-prob}
\end{align}
In the above chain of inequalities, $(a)$ holds because event $\mathcal{E}_2\cap\bar{\mathcal{E}}_1$ is a
subset of the event on the right hand side.

It is worthwhile to interpret inequality \eqref{jrnl1-onlinedecode-prob}. The error event of interest, roughly speaking, is partitioned into the following two events:
\begin{itemize}
\item[(1)] The auxiliary decoder makes an error and the non-unique decoder of
Section \ref{jrnl1-indirectcodedesign} also makes an error. 
\item[(2)] The auxiliary decoder makes an error but the non-unique decoder of
Section \ref{jrnl1-indirectcodedesign} decodes correctly. We will show
that the probability of this event is small. Note that under this error
event, (a) component decoder 1 fails (i.e., it is not possible to decode
$(M_0,M_{10})$ by treating $V_2^n$ as noise), but still (b) non-unique decoder
succeeds (i.e., the non-unique decoder must be deriving useful information
by considering $V_2^n$). By showing that this error event has a small
probability, we in effect show that whenever (a) and (b) hold, it is
possible to jointly uniquely decode the $V_2^n$ codeword as well. This
makes the rough intuition from Section~\ref{jrnl1_directdecoding} more
concrete.
\end{itemize} 

To bound the error probability, we bound the two terms of inequality \eqref{jrnl1-onlinedecode-prob} separately. 
The first term of \eqref{jrnl1-onlinedecode-prob} {is} bounded by the probability of the non-unique decoder making an error:
\begin{eqnarray}
 &&\hspace{-.7cm}\Pr\!\!\left( \hspace{-.3cm} \left.\begin{array}{c}  ( U^n(1,1),V_2^n(1,1,1),V_3^n(1,1,1),Y^n_2\!)\in A^n_\epsilon,\\\text{ and}\\(U^n(\tilde{m}_0,\tilde{s}_0),Y^n_2)\in A^n_\epsilon\\  \text{for some }(\tilde{m}_0,\tilde{s}_0)\neq(1,1)\\\text{ and, }\\(U^n(\hat{m}_0,\hat{s}_0),V^n_2(\hat{m}_0,\hat{s}_0,\hat{t}_2),Y^n_2)\in A^n_\epsilon \\ \text{for some }(\hat{m}_0,\hat{s}_0)\neq(1,1) \text{ and } \hat{t}_2\end{array}\right|\mathcal{I}=\mathbf{1}\hspace{-.15cm} \right)\nonumber\\
&\leq&
\!\Pr\left( \hspace{-.3cm}\left.\begin{array}{c}(U^n(\hat{m}_0,\hat{s}_0),V^n_2(\hat{m}_0,\hat{s}_0,\hat{t}_2),Y^n_2)\in A^n_\epsilon \\ \text{for some }(\hat{m}_0,\hat{s}_0)\neq(1,1) \text{ and } \hat{t}_2\end{array}\right|\mathcal{I}=\mathbf{1}\hspace{-.1cm}\right)\nonumber\\
&\leq&\hspace{-.6cm}\sum_{\substack{(\hat{m}_0,\hat{s}_0)\neq(1,1)\\\hat{t}_2}}\hspace{-.6cm}\Pr\!\left(\hspace{-.25cm} \left.\begin{array}{l}( U^n (\hat{m}_0,\hat{s}_0),V^n_2 (\hat{m}_0,\hat{s}_0,\hat{t}_2),Y^n_2\!)\!\in\! A^n_\epsilon\end{array}\!\!\!\right|\!\mathcal{I}=\mathbf{1} \right)\nonumber\\
&\leq& 2^{nT_2}2^{n(R_0+S_0)}2^{-n(I(UV_2;Y_2)-\gamma_1(\epsilon))}.\label{jrnl1-pe-term1}
\end{eqnarray}

The second term of \eqref{jrnl1-onlinedecode-prob} {is} upper-bounded by the expression in \eqref{jrnl1-pe-term2}, as we elaborate.

\begin{align}
&\hspace{-.3cm}\Pr\left(\hspace{-.3cm}\left.\begin{array}{c} ( U^n(1,1),V_2^n(1,1,1),V_3^n(1,1,1),Y^n_2\!)\!\in\! A^n_\epsilon,\\\text{ and}\\
(U^n(\tilde{m}_0,\tilde{s}_0),Y^n_2)\!\in\! A^n_\epsilon \\ \text{for some
}(\tilde{m}_0,\tilde{s}_0)\neq(1,1),\\\text{ and}\\ \text{all }
( U^n(\hat{m}_0,\hat{s}_0),V^n_2(\hat{m}_0,\hat{s}_0,\hat{t}_2),Y^n_2\!)\!\in\!
A^n_\epsilon\\ \text{ are s.t.}\ (\hat{m}_0,\hat{s}_0)=(1,1)\\\text{ with at
least one s.t. } \hat{t}_2\neq 1\end{array}\!\right|\mathcal{I}=\mathbf{1} \hspace{-.15cm}\right)\nonumber\\
&\leq2^{n(R_0+S_0+T_2)}2^{-n(I(UV_2;Y_2)-\gamma_2(\epsilon)-\delta(\epsilon))}\label{jrnl1-pe-term2}
\end{align}
We derive the bound \eqref{jrnl1-pe-term2} as follows. First, we write the following chain of inequalities
.
\allowdisplaybreaks
\begin{align}
\nonumber\\
&\Pr\left(\hspace{-.2cm}\left.\begin{array}{c} ( U^n (1,1),V_2^n (1,1,1),V_3^n (1,1,1),Y^n_2\!)\!\in\! A^n_\epsilon,\\\text{ and}\\
(U^n(\tilde{m}_0,\tilde{s}_0),Y^n_2)\!\in\! A^n_\epsilon \\ \text{for some
}(\tilde{m}_0,\tilde{s}_0)\neq(1,1),\\\text{ and}\\ \text{all }
( U^n(\hat{m}_0,\hat{s}_0),V^n_2(\hat{m}_0,\hat{s}_0,\hat{t}_2),Y^n_2\!)\!\in\!
A^n_\epsilon\\ \text{ are s.t.}\ (\hat{m}_0,\hat{s}_0)=(1,1)\\\text{ with at
least one s.t. } \hat{t}_2\neq 1\end{array}\!\right|\mathcal{I}=\mathbf{1} \!\! \right)\nonumber\\
&\leq\Pr\!\left(\hspace{-.3cm}\left.\begin{array}{c} ( U^n (1,1),V_2^n (1,1,1),V_3^n (1,1,1),Y^n_2\!)\!\in\! A^n_\epsilon,\\\text{ and}\\
(U^n(\tilde{m}_0,\tilde{s}_0),Y^n_2)\in A^n_\epsilon \\ \text{for some
}(\tilde{m}_0,\tilde{s}_0)\neq(1,1),\\\text{ and }\\
(U^n(1,1),V^n_2(1,1,\hat{t}_2),Y^n_2)\in
A^n_\epsilon\\ \text{for some
}\hat{t}_2\neq 1\end{array}\!\! \right|\!\mathcal{I}=\mathbf{1}\hspace{-.15cm}\right)\nonumber\\
&\leq \sum_{\substack{(\tilde{m}_0,\tilde{s}_0)\neq(1,1)\\\hat{t}_2\neq 1}}\hspace{-.55cm}\Pr \left(\hspace{-.3cm}\left.\begin{array}{c}  ( U^n ( 1, 1 ),\!V_2^n ( 1, 1, 1 )\hspace{1cm}\\\hspace{1cm},\!V_3^n (1 , 1, 1 ),\!Y^n_2\!)\! \in\!  A^n_\epsilon ,\\\text{ and}\\(U^n(\tilde{m}_0,\tilde{s}_0),Y^n_2)\in A^n_\epsilon\\\text{ and}\\ (U^n(1,1),V^n_2(1,1,\hat{t}_2),Y^n_2)\in A^n_\epsilon,\end{array}\!\!\!\right|\!\mathcal{I}\!=\!\mathbf{1}\hspace{-.15cm} \right)\nonumber\\
\nonumber\\
&\leq 2^{n(R_0 + S_0 + T_2)}\nonumber\\
&\quad\times\Pr\!\left(\hspace{-.3cm}\left.\begin{array}{c}  ( U^n ( 1, 1 ), V_2^n ( 1, 1, 1 )\hspace{1cm}\\\hspace{1cm}, V_3^n ( 1, 1, 1 ), Y^n_2\!)\!\in\! A^n_\epsilon ,\\\text{ and}\\(U^n(\tilde{m}_0,\tilde{s}_0),Y^n_2)\in A^n_\epsilon ,\\\text{ and}\\(U^n(1,1),V^n_2(1,1,\hat{t}_2),Y^n_2)\in A^n_\epsilon\end{array}\!\!\!\right|\!\mathcal{I}\!=\!\mathbf{1}\!\!\right)\label{jrnl1-ref3}
\end{align}
where we have $(\tilde{m}_0,\tilde{s}_0)\neq 1$ and $\hat{t}_2\neq 1$ in
the event in inequality \eqref{jrnl1-ref3}.

Next, we bound the probability term in \eqref{jrnl1-ref3}.
In what follows, $U^n$, $V_2^n$, $V_3^n$, $\tilde{U}^n$, $\hat{V}^n_2$ denote $U^n(1,1)$, $V_2^n(1,1,1)$, $V_3^n(1,1,1)$, $U^n(\tilde{m}_0,\tilde{s}_0)$, $V^n_2(1,1,\hat{t}_2)$,  respectively. Also, $p_{U^n|\mathcal{I}}(u^n|\mathbf{1})$ denotes $\Pr(U^n=u^n|\mathcal{I}=\mathbf{1})$. We sometimes drop the subscripts of probabilities if there is no ambiguity; e.g.,  $p(u^n|\mathbf{1})$ is just $p_{U^n|\mathcal{I}}(u^n|\mathbf{1})$.

In order to bound the probability term in \eqref{jrnl1-ref3}, one should treat $p_{U^nV_2^nV_3^nY^n_2\tilde{U}^n\hat{V}_2^n|\mathcal{I}}(u^n,v_2^n,v_3^n,y_2^n,\tilde{u}^n,\hat{v}_2^n|\mathbf{1})$. This would have been a straightforward task 
if the generated codebook was independent of indices $\mathcal{I}$. Nonetheless, it is an important observation that this is not the case\footnote{This was pointed out to us by anonymous reviewers, to whom we are grateful. Similar observations are made in \cite{MineroLimKim13} and \cite{GroverWagnerSahai13} where proof techniques were developed to handle such technicalities.}. For example, given $U^n$ (and under the conditioning $\mathcal{I}=\mathbf{1}$), $Y^n_2$ may not be independent of $\hat{V}^n_2$. Interestingly however, almost the same result holds. We address this in the following. We follow the proof idea in \cite{MineroLimKim13} to address this technicality. 
\begin{align}
&\Pr\left(\hspace{-.2cm}\left.\begin{array}{l} (U^n,V_2^n,V_3^n,Y_2^n)\in\mathcal{A}_\epsilon^n \text{ and }\\(\tilde{U}^n,Y^n_2)\in A^n_\epsilon\text{ and } (U^n,\hat{V}^n_2,Y^n_2)\in A^n_\epsilon\end{array}\right|\mathcal{I}=\mathbf{1}\right)\nonumber\\
&=\!\!\!\!\!\! \! \sum_{ (\!u^n\!\!,v_2^n\!,v_3^n\!,y_2^n\!) \in\mathcal{A}_\epsilon^n}  \!\!  \sum_{\substack{\tilde{u}^n:\\  (\!\tilde{u}^n\!\!,y_2^n\!) \in\mathcal{A}^n_\epsilon}}  \sum_{\substack{\hat{v}_2^n:\\ (\!{u}^n\!\!,\hat{v}_2^n\!,y_2^n\!) \in\mathcal{A}^n_\epsilon}} \!\! p(u^n\!\!, v_2^n\! , v_3^n\! , y_2^n\! , \tilde{u}^n\!\!, \hat{v}_2^n|\mathbf{1}\!)\nonumber\\
&=\!\!\!\!\!\!\!  \sum_{ (\!u^n\!\!,v_2^n\!,v_3^n\!,y_2^n\!) \in\mathcal{A}_\epsilon^n}  \!\! \sum_{\substack{\tilde{u}^n:\\  (\!\tilde{u}^n\!\!,y_2^n\!) \in\mathcal{A}^n_\epsilon}}  \sum_{\substack{\hat{v}_2^n:\\ (\!{u}^n\!\!,\hat{v}_2^n\!,y_2^n\!) \in\mathcal{A}^n_\epsilon}} \hspace{-.6cm} \begin{array}{l}\\\\\!\left[\,p (\!u^n\!\!, v_2^n \!, v_3^n\! , y_2^n|\mathbf{1}\!)\right.\nonumber\\\hspace{.15cm} \left.\times p (\!\tilde{u}^n|u^n\!\!, v_2^n\! , v_3^n\! , y_2^n\! , \mathbf{1}\!)\right.\nonumber\\\hspace{.3cm}  \left.\times p (\!\hat{v}_2^n|u^n\!\!, v_2^n\! , v_3^n\! , y_2^n\! , \tilde{u}^n\!\!, \mathbf{1}\!)\right]\end{array}\nonumber\\
&=\!\!\!\!\!\!\!  \sum_{ (\!u^n\!\!,v_2^n\!,v_3^n\!,y_2^n\!) \in\mathcal{A}_\epsilon^n} \!\!  \sum_{\substack{\tilde{u}^n:\\  (\!\tilde{u}^n\!\!,y_2^n\!) \in\mathcal{A}^n_\epsilon}}  \sum_{\substack{\hat{v}_2^n:\\ (\!{u}^n\!\!,\hat{v}_2^n\!,y_2^n\!) \in\mathcal{A}^n_\epsilon}} \hspace{-.5cm} \begin{array}{l}\\\\\!\left[\,p (\!{u}^n\!\! , v_2^n\! , v_3^n\! , y_2^n|\mathbf{1}\!)\right.\nonumber\\\hspace{.15cm} \left.\times p_{U^n} (\!\tilde{u}^n\!)\right.\nonumber\\\hspace{.3cm} \left.\times p (\!\hat{v}_2^n|{u}^n\!\!, v_2^n\! , v_3^n\! , y_2^n\! , \mathbf{1}\!)\right]\end{array}\nonumber\\
&\stackrel{(a)}{=}\!\!\!\!\!\!\!  \sum_{ (\!u^n\!\!,v_2^n\!,v_3^n\!,y_2^n\!) \in\mathcal{A}_\epsilon^n}   \!\! \sum_{\substack{\tilde{u}^n:\\  (\!\tilde{u}^n\!\!,y_2^n\!) \in\mathcal{A}^n_\epsilon}}  \sum_{\substack{\hat{v}_2^n:\\  (\!{u}^n\!\!,\hat{v}_2^n\!,y_2^n\!) \in\mathcal{A}^n_\epsilon}} \hspace{-.5cm} \begin{array}{l}\\\\\left[p (\!{u}^n\!\!  , v_2^n\! , v_3^n\! , y_2^n|\mathbf{1}\!)\right.\nonumber\\\hspace{.15cm} \left.\times p_{U^n} (\!\tilde{u}^n\!)\right.\nonumber\\\hspace{.3cm}  \left.\times p (\!\hat{v}_2^n|{u}^n\!\!, v_2^n\! , v_3^n\! , \mathbf{1}\!)\right]\end{array}\nonumber\\
&\stackrel{(b)}{\leq}\!\!\!\!\!\! \sum_{ (\!u^n\!\!,v_2^n\!,v_3^n\!,y_2^n\!) \in\mathcal{A}_\epsilon^n}   \!\! \sum_{\substack{\tilde{u}^n:\\  (\!\tilde{u}^n\!\!,y_2^n\!) \in\mathcal{A}^n_\epsilon}}  \sum_{\substack{\hat{v}_2^n:\\\\  (\!{u}^n\!\!,\hat{v}_2^n\!,y_2^n\!) \in\mathcal{A}^n_\epsilon}} \hspace{-.5cm} \begin{array}{l}\\\vspace{.2cm}\\\!\left[\,p (\!{u}^n\!\!, v_2^n\! , v_3^n\! , y_2^n|\mathbf{1}\!)\vphantom{\times 2^{n\delta(\epsilon)}p_{\hat{V}_2^n|U^n} (\!\hat{v}_2^n|{u}^n\!)}\right.\nonumber\\\hspace{.15cm} \left.\times p_{{U}^n} (\!\tilde{u}^n\!)\right.\nonumber\\\hspace{.3cm} \left.\times 2^{n\delta(\epsilon)}p_{\hat{V}_2^n|U^n} (\!\hat{v}_2^n|{u}^n\!)\right]\end{array}\nonumber\\
\nonumber\\
&\leq2^{n\delta(\epsilon)}2^{-n(I(U;Y_2)-\gamma^\prime(\epsilon))}2^{-n(I(V_2;Y_2|U)-\gamma^\prime(\epsilon))}\nonumber\\&\quad\times \sum_{(u^n\!,v_2^n,v_3^n,y_2^n)\in\mathcal{A}_\epsilon^n}\!\!\!\!\begin{array}{l}p(u^n,v_2^n,v_3^n,y_2^n|\mathbf{1})\end{array}\nonumber\\
&\leq 2^{-n(I(UV_2;Y_2)-\gamma_2(\epsilon)-\delta(\epsilon))}\nonumber
\end{align}
Step $(a)$ follows from the fact that $\hat{V}_2^n-U^n,V_2^n,V_3^n,\mathcal{I}-Y_2^n$ forms a Markov chain. In order to prove step (b), we show that  conditioned on $U^n$, $\hat{V}^n_2$ is ``almost" independent of $V_2^n,V_3^n,\mathcal{I}$. More precisely, we use similar steps as in \cite[Lemma 1]{MineroLimKim13} and show in Appendix \ref{app:typicalitylemma} that for any jointly typical tuple $(u^n,v_2^n,v_3^n)\in\mathcal{A}_\epsilon^n$ and any {$\epsilon>0$}, there is a large enough $n$ such that
$p(\hat{v}^n_2|u^n,v_2^n,v_3^n,\mathbf{1})\leq 2^{n\delta(\epsilon)}p(\hat{v}^n_2|u^n)$, where $\delta(\epsilon)$ tends to zero as $\epsilon\to 0$.

We conclude the error probability analysis by putting together inequalities \eqref{jrnl1-pe-bound}, \eqref{jrnl1-onlinedecode-prob}, \eqref{jrnl1-pe-term1}, and \eqref{jrnl1-pe-term2} to obtain {that the error probability at the auxiliary decoder is bounded as in inequality \eqref{jrnl1-indirecterror}.}
So for large enough $n$, the auxiliary decoder succeeds with high probability
 if the non-unique decoding constraint \eqref{jrnl1_indirectY2} is satisfied{; i.e., when the non-unique decoder succeeds with high probability}.

One can now argue that if the auxiliary decoder succeeds with high
probability for an operating point, then there also exists a joint unique decoding
scheme that succeeds with high probability.  The idea is that for all
operating points (except in a subset of the rate region of measure zero),
each of the two component (joint unique) decoders $1$ and $2$ have either a high or
a low probability of success.  So, if the operating point is such that the
auxiliary decoder decodes correctly with high probability, then at least
one of the component decoders should also decode correctly with high
probability, giving us the joint unique decoding scheme we were looking for. This
is summarized in Lemma \ref{jrnl1-Lemmaexistence}, and the reader is
referred to Appendix \ref{app:Lemmaexistence} for the proof.

\begin{lemma}
\label{jrnl1-Lemmaexistence}
Given any operating point (except in a subset of the rate region of measure
zero), if the auxiliary decoder succeeds with high probability under the random coding
experiment, then there exists a joint unique decoding scheme that also succeeds
with high probability.
\end{lemma}

A similar argument goes through for receiver $Y_3$. The random coding
argument for the joint unique decoding scheme can now be completed as usual.

\subsection{Discussion}
\begin{remark}
In Sections \ref{jrnl1_directdecoding} and \ref{onlinedecoder}, we did not
consider cases where $R_0+S_0=I(U;Y_2)$ or $R_0+S_0=I(U;Y_3)$ (i.e., a subset
of measure zero). This is enough since we may get arbitrarily close to such points. 
\end{remark}
\begin{remark}
\label{jrnl-remarkmessagestructure}
 In Sections \ref{jrnl1_directdecoding} and \ref{onlinedecoder}, we fixed
the encoding scheme to be that of \cite{NairElGamal09}. The message
splitting and the structure of the codebook is therefore a priori assumed
to be that of \cite{NairElGamal09}, even when $R_0+S_0< I(U;Y_2)$ and
message $M_{12}$ is not jointly decoded at $Y_2$. However, in such cases
this extra message structure is not required and one can consider message
$M_{12}$ as a part of message $M_{11}$.
\end{remark}

%% file: examplesrevision.tex
\section{More examples}
\label{jrnl1-examples}

We saw that joint unique decoding was sufficient to achieve the inner-bound of
\cite{NairElGamal09}. This is not coincidental and the same phenomenon can
be observed for example in the work of Chong, Motani, Garg and El~Gamal
\cite{ChongMotaniGargElGamal08} where the region obtained by non-unique
decoding turned out to be equivalent to that of Han and Kobayashi in
\cite{HanKobayashi81}. Similarly for noisy network coding \cite{NNC11}, it has been shown that the same rate region can be obtained employing joint unique decoding \cite{WuXie10, WuXieAllerton10, KramerITW11, HouKramerISIT12}.
It was also observed in \cite{ZaidiPiantanidaShamai12} that non-unique decoding is not essential to achieve the capacity region of certain state-dependent multiple access channels and joint unique decoding suffices.
Non-unique decoding schemes have appeared also in
\cite{NairElGamal10,BandemerElGamal11,ChiaElGamal11}. We consider these
three problems next and show that employing joint unique decoders, one can achieve the same proposed inner-bounds. To show such equivalence, we use the proof technique that we developed in Section \ref{onlinedecoder}


\subsection{Two-receiver compound channel with state noncausally available at the encoder}
\label{statenoncausally}
An inner-bound to the common message capacity region of a 2-receiver compound channel with discrete memoryless state noncausally available at the encoder is derived in \cite{NairElGamal10}. The inner-bound is established
using superposition coding, Marton's coding, and non-unique decoding schemes. { More precisely, the achievable scheme is as follows:
\subsubsection{Codebook generation}
Fix $p_{WUV}(w,u,v)$ and $f(w,u,v,s)$. For each message $m$, generate randomly and independently $2^{nT_0}$ sequences $W^n(m,l_0)$  according to $\prod_i p_W(w_i)$. For each $(m,l_0)$, generate  randomly and conditionally independently (i) $2^{nT_1}$ sequences $U^n(m,l_0,l_1)$  according to $\prod_i p_{U|W}(u_i|w_i)$ and (ii)  $2^{nT_2}$ sequences $V^n(m,l_0,l_2)$ according to $\prod_i p_{V|W}(v_i|w_i)$. 
\subsubsection{Encoding} Given message $m$ and state $s^n$, the encoder finds $l_0$ such that $(W^n(m,l_0),s^n)\in\mathcal{A}^n_\epsilon$. If there is more than one such index,  one is chosen uniformly at random\footnote{We allow a small modification to \cite{NairElGamal10} in randomly choosing the index $l_0$ whenever there is not a unique choice.}. If there is no such index, a random index is chosen among $\{1,\ldots,2^{nT_0}\}$. Next, the encoder finds $l_1$ and $l_2$ such that $(W^n(m,l_0),s^n,U^n(m,l_0,l_1),V^n(m,l_0,l_2))\in \mathcal{A}^n_\epsilon$.  If there is more than one such index pair, one pair is chosen uniformly at random. If there is none, a random index pair $(l_1,l_2)$ is chosen among $\{1\ldots,2^{nT_1}\}\times\{1,\ldots,2^{nT_2}\}$. The encoder transmits $x^n$, $x_i=f(w_i,u_i,v_i,s_i)$, where $w^n=W^n(m,l_0)$, $u^n=U^n(m,l_0,l_1)$, and $v^n=V^n(m,l_0,l_2)$.
\subsubsection{Decoding}
Receiver $Y_1$ declares message $M$ to be the unique index
$m$ for which $(W^n(m,l_0 ),U^n (m,l_0,l_1 ),Y^n_1 )$ is jointly typical for some $l_0\in\{1,\ldots,2^{nT_0}\}$ and $l_1\in\{1,\ldots,2^{nT_1}\}$. Receiver $Y_2$ follows a similar scheme. 

In this problem,
we show that employing joint unique decoders lets us achieve the same inner-bound of Theorem $1$ of \cite{NairElGamal10}. We outline the proof which is built on the proof technique of Subsection \ref{onlinedecoder}.
Define the auxiliary decoder (at receiver $Y_1$) to have access to two component (joint unique) decoders: one jointly uniquely decoding indices
$m_0,l_0$, and one jointly uniquely decoding indices $m_0,l_0,l_1$. The auxiliary decoder declares an error if either (a) both component decoders declare an error or (b) neither of them declare an error but they do not agree on their decoded $m_0$ and $l_0$ indices.

We now analyze the error probability. Assume, without any loss of generality, that the originally sent indices were $(m,l_0,l_1,l_2)=(1,1,1,1)$. We denote this event by $\mathcal{I}=\mathbf{1}$.  Proceeding as in Section \ref{onlinedecoder}, the error probability of the auxiliary decoder is bounded  by the following probability term.
\begin{align}
&\Pr(\text{error}|\mathcal{I}=\mathbf{1})\nonumber\\
&\leq\epsilon + \Pr \left(    \left.\begin{array}{c} ( W^n ( 1, 1 ), S^n,U^n ( 1, 1, 1 )\hspace{1cm}\\\hspace{1cm}, V^n ( 1 , 1 , 1 ),Y_1^n ) \in \mathcal{A}_\epsilon^n ,\\ \text{and}\\(W^n(\tilde{m},\tilde{l}_0),Y^n_1)\in A^n_\epsilon\\  {\text{for some }(\tilde{m},\tilde{l}_0)\neq(1,1),}\\ \text{and}\\ (W^n(\hat{m},\hat{l}_0),U^n(\hat{m},\hat{l}_0,\hat{l}_1),Y^n_1)\in A^n_\epsilon\\  {\text{for some }(\hat{m},\hat{l}_0,\hat{l}_1)\neq(1,1,1)}\end{array}    \right|\mathcal{I}=\mathbf{1} \! \right)\nonumber\\\label{jrnl-example2}
\end{align}
The probability term on the right hand side of inequality \eqref{jrnl-example2} is very similar to what we obtained in inequality  \eqref{jrnl1-errorauxiliary} and is analyzed in the same manner (with the subtle difference that $W^n(m,l_0)$ is indexed not only by the message but also by the state, which asks for a more careful treatment). See Appendix~\ref{appendixnoncausally}. 
We follow similar steps to conclude that the auxiliary decoder performs reliably under the \textit{non-unique decoding constraints} of \cite{NairElGamal10}. 
So, 
there exists a joint unique decoding scheme that performs reliably under those decoding constraints. More explicitly,
the proposed joint unique decoding scheme would be
joint unique decoding of $m$ and $l_0$, if $R_0+T_0<I(W;Y_1)$; and joint unique decoding of $m$, $l_0$ and $l_1$, otherwise.}

\subsection{Three-user deterministic interference channel}
In \cite{BandemerElGamal11}, an inner-bound to the capacity region of a
class of deterministic interference channels with three user pairs is
derived. The key idea is to simultaneously decode the combined interference
signal and the intended message at each receiver and this is done by a non-unique decoding scheme. 
{ We focus on Theorem $1$ of \cite{BandemerElGamal11} and
to have the paper self contained we briefly mention the encoding and decoding scheme. The deterministic interference channels that are considered here are described by the following deterministic relations between the inputs and the outputs\footnote{All results easily generalize to interference channels with noisy observations (e.g.,  \cite[Theorem 4]{BandemerElGamal11}).}:  $Y_k=f_k(X_{kk},S_k)$ where $S_1=h_1(X_{21},X_{31})$, $S_2=h_2(X_{12},X_{32})$, an $S_3=h_3(X_{23},X_{13})$ and $X_{lk}=g_{lk}(X_l)$ for every $l,k\in\{1,2,3\}$. It is assumed that functions $h_k$ and $f_k$ are one-to-one mappings when either of their arguments is fixed.
\paragraph*{Codebook generation}
Fix the probability mass function (pmf) $p(q)p(x_1|q)p(x_2|q)p(x_3|q)$. Sequence $Q^n$ is generated according to $\prod_i p_Q(q_i)$. For each $k=1,2,3$,  sequences $X_k^n(m_k)$, $m_k\in\{1,\ldots,2^{nR_k}\}$, are generated randomly and conditionally independently according to $\prod_i p_{X_k|Q}(x_{k,i}|q_i)$. 
\paragraph*{Encoding} To send message $m_k$, transmitter $k$ transmits $X_k^n(m_k)$. 
\paragraph*{Decoding}  Upon receiving $Y^n_1$, 
decoder $1$ declares that $m_1$
is sent if it is the unique message such that
$$(Q^n,X^n_1(m_1),S^n_1(m_2,m_3),X^n_{21}(m_2),X^n_{31}(m_3),Y^n_1)\in\mathcal{A}^n_\epsilon$$  for some $m_2\in[1:2^{nR_2}]$ and $m_3\in[1:2^{nR_3}]$. Decoders $2$ and $3$ work similarly.}

Here, we use the proof technique of Section~\ref{onlinedecoder} to prove
that a code design that employs joint unique decoders achieves the same
inner-bound.
  
Define the auxiliary decoder (at receiver $Y_1$) to have access to four
component (joint unique) decoders: one jointly uniquely decoding $X^n(m_1)$, one
jointly uniquely decoding $X^n_1(m_1)$ and $X^n_{21}(m_2)$, one jointly uniquely
decoding $X^n_1(m_1)$ and $X^n_{31}(m_3)$ and finally one jointly uniquely
decoding all sequences $X^n(m_1)$, $X^n_{21}(m_{2})$, $X^n_{31}(m_3)$, and
$S_1^n(m_2,m_3)$. The auxiliary decoder declares an error if either (a) all
component decoders declare error, or (b) not all of the decoders that
decode without declaring an error agree on the decoded index $m_0$ (i.e.,
among those component decoders that do not declare an error, there is not a
common agreement on the decoded index $m_0$).

We now analyze the error probability of the auxiliary decoder. We assume without any loss of generality that $(m_1,m_2,m_3)=(1,1,1)$ was sent. Proceeding as in Section \ref{onlinedecoder}, the error probability of the auxiliary decoder is bounded by
inequality \eqref{jrnl1-peexample3} as follows.

\begin{align}
&\Pr(\text{error})\nonumber\\
&\leq \epsilon + \Pr \left(          \left.\begin{array}{c}(Q^n ,X^n_1(\bar{m}_1),Y^n_1)\in\mathcal{A}^n_\epsilon\\{\text{ for some }\bar{m}_1\neq 1},\text{ and}\\ (Q^n ,X^n_1(\breve{m}_1),X^n_{21}(\breve{m}_2),Y^n_1)\in\mathcal{A}^n_\epsilon\\{\text{ for some }(\breve{m}_1,\breve{m}_2)\neq (1,1)},\text{ and}\\ (Q^n ,X^n_1(\dot{m}_1),X^n_{31}(\dot{m}_3),Y^n_1)\in\mathcal{A}^n_\epsilon\\{\text{ for some }(\dot{m}_1,\dot{m}_3)\neq (1,1)},\text{ and}\\\left( Q^n ,X^n_1(\hat{m}_1),S^n_1(\hat{m}_2,\hat{m}_3)\right. \hspace{1cm}\ \\\  \hspace{.4cm}\left.,X^n_{21}(\hat{m}_2),X^n_{31}(\hat{m}_3),Y^n_1 \right) \in \mathcal{A}^n_\epsilon\\{\text{ for some}\,(\hat{m}_1,\hat{m}_2,\hat{m}_3) \neq  (1,1,1)}\end{array}    \right| \mathcal{I} = \mathbf{1}   \right)\nonumber\\
&\leq \epsilon + \Pr \left(          \left.\begin{array}{c}(Q^n ,X^n_1(\bar{m}_1),Y^n_1)\in\mathcal{A}^n_\epsilon\\{\text{ for some }\bar{m}_1\neq 1},\text{ and}\\ (Q^n ,X^n_1(\breve{m}_1),X^n_{21}(\breve{m}_2),Y^n_1)\in\mathcal{A}^n_\epsilon\\{\text{ for some }(\breve{m}_1,\breve{m}_2)\neq (1,1)},\text{ and}\\ (Q^n ,X^n_1(\dot{m}_1),X^n_{31}(\dot{m}_3),Y^n_1)\in\mathcal{A}^n_\epsilon\\{\text{ for some }(\dot{m}_1,\dot{m}_3)\neq (1,1)},\text{ and}\\\left( Q^n ,X^n_1(\hat{m}_1),S^n_1(\hat{m}_2,\hat{m}_3)\right. \hspace{1cm}\ \\\  \hspace{.4cm}\left.,X^n_{21}(\hat{m}_2),X^n_{31}(\hat{m}_3),Y^n_1 \right) \in \mathcal{A}^n_\epsilon\\{\text{ for some}\,(\hat{m}_2,\hat{m}_3),\ \hat{m}_1\neq1}    \end{array}    \right| \mathcal{I} = \mathbf{1}   \right)\nonumber\\
&\quad+\Pr \left(          \left.\begin{array}{c}(Q^n ,X^n_1(\bar{m}_1),Y^n_1)\in\mathcal{A}^n_\epsilon\\{\text{ for some }\bar{m}_1\neq 1},\text{ and}\\ (Q^n ,X^n_1(\breve{m}_1),X^n_{21}(\breve{m}_2),Y^n_1)\in\mathcal{A}^n_\epsilon\\{\text{ for some }(\breve{m}_1,\breve{m}_2)\neq (1,1)},\text{ and}\\ (Q^n ,X^n_1(\dot{m}_1),X^n_{31}(\dot{m}_3),Y^n_1)\in\mathcal{A}^n_\epsilon\\{\text{ for some }(\dot{m}_1,\dot{m}_3)\neq (1,1)},\text{ and}\\\left(Q^n ,X^n_1(\hat{m}_1),S^n_1(\hat{m}_2,\hat{m}_3)\right. \hspace{1cm}\ \\\  \hspace{.4cm}\left.,X^n_{21}(\hat{m}_2)X^n_{31}(\hat{m}_3),Y^n_1\right) \in \mathcal{A}^n_\epsilon\\{\text{ for some}\,(\hat{m}_2,\hat{m}_3) \neq  (1,1),\, \hat{m}_1 = 1}\end{array}    \right| \mathcal{I} = \mathbf{1}   \right)\label{jrnl1-peexample3}
\end{align}

As before, the first probability term of inequality
\eqref{jrnl1-peexample3} is upper-bounded by the probability of an indirect
decoder making an error; { i.e.}, by the expression below.
\begin{align}
\Pr  \left(     \begin{array}{l|l}( Q^n\hspace{-.1cm}, X^n_1( \hat{m}_1 ), S^n_1( \hat{m}_2, \hat{m}_3 ) \hspace{1cm}\ \\\  \hspace{1cm} , X^n_{21}( \hat{m}_2 ), X^n_{31}( \hat{m}_3 ), Y^n_1 ) \in \mathcal{A}^n_\epsilon   \\ \text{
for some }(\hat{m}_2,\hat{m}_3)\text{ and }
\hat{m}_1\neq1\end{array}\, \mathcal{I}  =  \mathbf{1}   \right)\label{jrnl1-indirecterrorexample3}
\end{align}
In \cite{BandemerElGamal11}, constraints on rates have been derived
under which this error probability approaches $0$ as $n$ grows large and an
achievable rate region has been characterized. We refer to these
constraints as the \textit{non-unique decoding constraints} of
\cite{BandemerElGamal11}.  One can show that under these decoding
constraints, the second probability term can also be made arbitrarily small
by choosing a sufficiently large $n$ (Appendix \ref{jrnl1-appendixlast}).  It then becomes clear that the auxiliary
decoder succeeds with high probability if the non-unique decoding
constraints of \cite{BandemerElGamal11} are satisfied. So, analogous to
Section~\ref{onlinedecoder}, we conclude that there exists a joint unique decoding
scheme that achieves the same inner-bound of Theorem $1$ in
\cite{BandemerElGamal11}.

\subsection{Three-receiver broadcast channel with common and confidential
messages}
\label{jrnl1-example1}

In \cite{ChiaElGamal11} a general 3-receiver broadcast channel with one
common and one confidential message set is studied. Inner-bounds and
outer-bounds are derived for the capacity regions under two setups of this
problem: when the confidential message is intended for one of the receivers
and when the confidential message is intended for two of the receivers.  We
only address the first setup here, and in particular Theorem~2 of
\cite{ChiaElGamal11}. The other inner-bounds can be similarly dealt with.
In Theorem~2, the authors establish an inner-bound to the secrecy capacity
region using the ideas of superposition coding, Wyner wiretap channel
coding, and non-unique decoding. We briefly explain the achievable scheme.
{ {\paragraph*{Codebook construction} Fix $p_{UV_0V_1V_2X}(u,v_0,v_1,v_2,x)$. Choose $R_r\geq 0$ such that $R_1-R_e+R_r\geq I(V_0;Z|U)+\delta(\epsilon)$. Randomly and independently generate $2^{nR_0}$ sequences $u^n(m_0)$,  each according to $\prod_i P_U(u_i)$. For each $m_0$, randomly and conditionally independently generate sequences $v_0^n(m_0,m_1,m_r)$, $(m_1,m_r)\in[1:2^{n(R_1+R_r)}]$, each according to $\prod_i P_{V_0|U(v_{0i}|u_i)}$. For each $(m_0,m_1,m_r)$: (i) generate sequences $v_1^n(m_0,m_1,m_r,t_1)$, $t_1\in\{1,\ldots,2^{nT_1}\}$, each according to $\prod_i p_{V_1|V_0(v_{1i}|v_{0i})}$, and partition the set $\{1,\ldots,2^{nT_1}\}$ into $2^{n\tilde{R}_1}$ equal size bins $\mathcal{B}(m_0,m_1,m_r,l_1)$, (ii) generate sequences $v_2^n(m_0,m_1m_r,t_2)$, $t_2\in\{1,\ldots,2^{nT_2}\}$, each according to the product distribution $\prod_i p_{V_2|V_0}(v_{2i}|v_{0i})$ and partition the set $\{1,\ldots,2^{nT_2}\}$ into $2^{n\tilde{R}_2}$ equal size bins $\mathcal{B}(m_0,m_1,m_r,l_2)$. For each product bin $\mathcal{B}(l_1)\times\mathcal{B}(l_2)$, find a jointly typical sequence pair $(v_1^n(m_0,m_1,m_r,t_1(l_1)),v_2^n(m_0,m_1,m_r,t_2(l_2))$, and associate it to the product bin. If there is more than one pair, one of the jointly typical pairs is picked uniformly at random. If there is no such pair, one pair is picked uniformly at random from the set of all possible pairs. Finally, for all $(m_0,m_1,m_r)$ and all their associated sequence pairs $(v_1^n(m_0,m_1,m_r,t_1(l_1)),v_2^n(m_0,m_1,m_r,t_2(l_2))$ a codeword $X^n(m_0,m_1,m_r, t_1(l_1),t_2(l_2))$ is generated according to $\prod_i p_{X|V_0V_1V_2}(x_i|v_{0i},v_{1i},v_{2i})$.
\paragraph*{Encoding} To send the message pair $(m_0,m_1)$, the encoder  chooses a random index $m_r\in\{1,\ldots, 2^{nR_r}\}$ and thus the sequence pair $(u^n(m_0),v_0^n(m_0,m_1,m_r))$. It then chooses a product bin index $(L_1,L_2)$ at random and selects the corresponding jointly typical pair  $(v_1^n(m_0,m_1,m_r,t_1(L_1)),v_2^n(m_0,m_1,m_r,t_2(L_2))$ in it. Finally the corresponding codeword $X^n(m_0,m_1,m_r, t_1(L_1),t_2(L_2))$ is sent.
\paragraph*{Decoding} Both legitimate
receivers $Y_1$ and $Y_2$ decode their messages of interest, $M_0$ and
$M_1$, by non-unique decoding schemes. More precisely, receiver $Y_1$ looks for the unique
triple $(m_0,m_1,m_r)$ such that the tuple
$(U^n(m_0),V^n_0(m_0.m_1,m_r),V^n_1(m_0,m_1,m_r,t_1),Y^n_1)$ is jointly
typical for some $t_1\in[1:2^{nT_1}]$. Receiver $Y_2$ follows a similar
scheme. Receiver $Z$ decodes $m_0$ directly by finding the jointly typical pair $(U^n(m_0),Z^n)$.}}

We use the proof technique of Subsection \ref{onlinedecoder}
to show that a code design that employs joint unique decoders achieves
the same inner-bound. To do so, we first present an auxiliary decoder which
succeeds with high probability under the decoding constraints of
\cite{ChiaElGamal11}, and then conclude that there exists a joint unique decoding
scheme that succeeds with high probability.

Define the auxiliary decoder (at receiver $Y_1$) to have access to two
component (joint unique) decoders, one jointly uniquely decoding indices
$m_0,m_1,m_r$ and the other jointly uniquely decoding indices
$m_0,m_1,m_r,t_1$. The auxiliary decoder declares an error if either (a)
both component decoders declare errors, or (b) if both of them decode and
their declared $(m_0,m_1,m_r)$ indices do not match. In all other cases it
declares the index triple $(m_0,m_1,m_r)$ according to the output of the
component decoder which did not declare an error. Proceeding as in
Section~\ref{onlinedecoder}, the error probability of
the auxiliary decoder can be bounded  by \eqref{jrnl1-boundauxiliary} as follows.{  As before, we assume without any loss of generality that the all-1-indices are chosen at the encoding stage, and we denote this event by $\mathcal{I}=\mathbf{1}$.
\begin{align}
&\Pr(\text{error}\;|\;\mathcal{I}=\mathbf{1})\nonumber\\
&  \leq    \epsilon  +  
\Pr\left( \!\!\!\!   
\left.    \begin{array}{c}\left(U^n(1),V_0^n(1,1,1),V_1^n(1,1,1,1)\right. \hspace{.8cm}\\ \hspace{2.2cm}\left.,V_2^n(1,1,1,1),Y^n_1\right) \!\in\! \mathcal{A}_\epsilon^n\\\text{ and}\\
        (U^n(\tilde{m}_0),V^n_0(\tilde{m}_0,\tilde{m}_1,\tilde{m}_r),Y^n_1) \!\in\!  A^n_\epsilon
           \\ {\text{for some }(\tilde{m}_0,\tilde{m}_1,\tilde{m}_r) \!\neq\! (1,1,1)}\\\text{ and}\\
        \left(U^n(\hat{m}_0),V^n_0(\hat{m}_0,\hat{m}_1,\hat{m}_r)\right. \hspace{.8cm}\\ \hspace{.8cm}\left.,V^n_1(\hat{m}_0,\hat{m}_1,\hat{m}_r,\hat{t}_1),Y^n_1\right)
             \!\in\!  A^n_\epsilon\\ {\text{for some }(\hat{m}_0,\hat{m}_1,\hat{m}_r,\hat{t}_1)
             \!\neq\! (1,1,1,1)}
    \end{array} \!\!  \right| \mathcal{I}\! =\! \mathbf{1} \! \right)\nonumber\\
&  \leq    \epsilon  + \Pr\left(    \!\!\!\!
   \left. \begin{array}{c} \left(U^n(1),V_0^n(1,1,1),V_1^n(1,1,1,1)\right. \hspace{.8cm}\\ \hspace{2.2cm}\left.,V_2^n(1,1,1,1),Y^n_1\right)\!\in\!\mathcal{A}_\epsilon^n\\\text{ and}\\
        (U^n(\tilde{m}_0),V^n_0(\tilde{m}_0,\tilde{m}_1,\tilde{m}_r),Y^n_1) \!\in\!  A^n_\epsilon
             \\ {\text{for some }(\tilde{m}_0,\tilde{m}_1,\tilde{m}_r) \!\neq\! (1,1,1)}\\\text{ and}\\
       \left(U^n(\hat{m}_0),V^n_0(\hat{m}_0,\hat{m}_1,\hat{m}_r)\right. \hspace{.8cm}\\ \hspace{.8cm}\left.,V^n_1(\hat{m}_0,\hat{m}_1,\hat{m}_r,\hat{t}_1),Y^n_1\right)
             \!\in\!  A^n_\epsilon \\ {\text{for some }(\hat{m}_0,\hat{m}_1,\hat{m}_r) \!\neq\! (1,1,1),\ \hat{t}_1}
    \end{array}\!\!\right| \mathcal{I}\!=\!\mathbf{1} \! \right)\nonumber\\
&       \hspace{.05cm}\; \; + \Pr\left(     \!\!\!\!
\left.        \begin{array}{c} \left(U^n(1),V_0^n(1,1,1),V_1^n(1,1,1,1)\right. \hspace{.8cm}\\ \hspace{2.2cm}\left.,V_2^n(1,1,1,1),Y^n_1\right)\!\in\!\mathcal{A}_\epsilon^n\\\text{ and}\\
            (U^n(\tilde{m}_0),V^n_0(\tilde{m}_0,\tilde{m}_1,\tilde{m}_r),Y^n_1) \!\in\!  A^n_\epsilon
                 \\ {\text{for some }(\tilde{m}_0,\tilde{m}_1,\tilde{m}_r) \!\neq\! (1,1,1)}\\\text{ and} \\
            \left(U^n(\hat{m}_0),V^n_0(\hat{m}_0,\hat{m}_1,\hat{m}_r)\right. \hspace{.8cm}\\ \hspace{.8cm}\left.,V^n_1(\hat{m}_0,\hat{m}_1,\hat{m}_r,\hat{t}_1),Y^n_1\right)
                 \!\in\!  A^n_\epsilon \\ {\text{for }(\hat{m}_0,\hat{m}_1,\hat{m}_r)\!=\!(1,1,1),\ 
                \hat{t}_1 \neq  1}\end{array} \!\!\right| \mathcal{I}\!=\!\mathbf{1}\!\right)\nonumber\\ 
&  \stackrel{(a)}{\leq}     \epsilon+2^{n(R_0+R_1+T_1+R_r-I(UV_0,V_1;Y_1)+\gamma_1(\epsilon))}\nonumber\\
     & \quad  +2^{n(R_1+T_1+R_r-I(V_0V_1;Y_1|U)+\gamma^\prime_1(\epsilon))}\nonumber\\
&    \quad   +\;  2^{n(R_0+R_1+T_1+R_r-I(UV_0,V_1;Y_1)+\gamma_2(\epsilon)+\delta(\epsilon))}\nonumber\\
      & \quad +\;2^{n(R_1+T_1+R_r-I(V_0V_1;Y_1|U)+\gamma^\prime_2(\epsilon)+\delta(\epsilon))}\label{jrnl1-boundauxiliary}
\end{align}
Here, $\gamma_1(\epsilon),\gamma^\prime_1(\epsilon),\gamma_2(\epsilon),\gamma^\prime_2(\epsilon),\delta(\epsilon)$ all go to zero as $\epsilon\to 0$.
To prove the inequality in step $(a)$, we bound each probability
term separately. 

The first probability term above is upper-bounded by the probability of a
non-unique decoder making an error (i.e., $2^{n(R_0+R_1+T_1+R_r-I(UV_0V_1;Y_1)+\gamma_1(\epsilon))}
        +2^{n(R_1+T_1+R_r-I(V_0V_1;Y_1|U)+\gamma_1^\prime(\epsilon))}$). This non-unique decoder is analyzed in \cite{ChiaElGamal11}
and shown to be reliable under the following two constraints to which we refer as the
\textit{non-unique decoding constraints} of \cite{ChiaElGamal11}.
\begin{eqnarray}
&&R_0+R_1+T_1+R_r<I(UV_0V_1;Y_1)-\gamma_1(\epsilon)\\
&&R_1+T_1+R_r<I(V_0V_1;Y_1|U)-\gamma^\prime_1(\epsilon)
\end{eqnarray}

The second term is upper-bounded by
further splitting the event and following steps similar to that of Subsection~\ref{onlinedecoder}. 

{{
\begin{align}
&\Pr  \left(\left. \begin{array}{c} \left( U^n ( 1 ), V_0^n ( 1, 1, 1 ), V_1^n ( 1, 1, 1, 1 )\right. \hspace{1cm}\\ \hspace{2.2cm}\left., V_2^n ( 1, 1, 1, 1 ), Y^n_1 \right) \!\in\! \mathcal{A}_\epsilon^n ,\\\text{ and}\\(U^n(\tilde{m}_0),V^n_0(\tilde{m}_0,\tilde{m}_1,\tilde{m}_r),Y^n_1) \!\in\!  A^n_\epsilon \\ {\text{for some }(\tilde{m}_0,\tilde{m}_1,\tilde{m}_r) \!\neq\! (1,1,1)}\\\text{ and} \\\left( U^n ( \hat{m}_0 ), V^n_0 ( \hat{m}_0 , \hat{m}_1 , \hat{m}_r ) \right. \hspace{1cm}\\ \hspace{1cm}\left., V^n_1 ( \hat{m}_0 , \hat{m}_1 , \hat{m}_r , \hat{t}_1 ), Y^n_1 \right)  \!\in\!   A^n_\epsilon \\ {\text{for some }(\hat{m}_0,\hat{m}_1,\hat{m}_r) \!=\! (1,1,1),\ \hat{t}_1 \!\neq\!  1}\end{array}    \right|\hspace{-.04cm}\mathcal{I} \!=\! \mathbf{1} \right)\nonumber\\
&       \leq        \Pr  \left(\!\! \left.\begin{array}{c} \left( U^n ( 1 ), V_0^n ( 1, 1, 1 ), V_1^n ( 1, 1, 1, 1 )\right. \hspace{1.2cm}\\ \hspace{2cm}\left., V_2^n ( 1, 1, 1, 1 ), Y^n_1 \right)  \!\in\!  \mathcal{A}_\epsilon^n ,\\\text{ and}\\(U^n(\tilde{m}_0),V^n_0(\tilde{m}_0,\tilde{m}_1,\tilde{m}_r),Y^n_1) \!\in\!  A^n_\epsilon \\ {\text{for some }(\tilde{m}_0,\tilde{m}_1,\tilde{m}_r ) \!\neq\! (1,1,1),\ \tilde{m}_0 \!\neq\! 1}\\ \text{ and}\\\left( U^n ( \hat{m}_0 ), V^n_0 ( \hat{m}_0 , \hat{m}_1 , \hat{m}_r )\right. \hspace{1cm}\\ \hspace{1cm}\left.,V^n_1 ( \hat{m}_0 , \hat{m}_1 , \hat{m}_r , \hat{t}_1 ), Y^n_1 \right) \! \in\!   A^n_\epsilon \\ {\text{for some }(\hat{m}_0,\hat{m}_1,\hat{m}_r) \!=\! (1,1,1),\ \hat{t}_1 \!\neq\!  1}\end{array} \!   \right| \mathcal{I} \!=\! \mathbf{1} \!\right)\nonumber\\
&  +      \Pr  \left( \!\!\left.\begin{array}{c} \left( U^n ( 1 ), V_0^n ( 1, 1, 1 ), V_1^n ( 1, 1, 1, 1 )\right. \hspace{1.2cm}\\ \hspace{2cm}\left., V_2^n ( 1, 1, 1, 1 ), Y^n_1 \right) \! \in\!  \mathcal{A}_\epsilon^n ,\\\text{ and}\\(U^n(\tilde{m}_0),V^n_0(\tilde{m}_0,\tilde{m}_1,\tilde{m}_r),Y^n_1)\! \in\!  A^n_\epsilon \\ {\text{for some }(\tilde{m}_0,\tilde{m}_1,\tilde{m}_r ) \!\neq\! (1,1,1),\ \tilde{m}_0\!=\!1}\\ \text{ and}\\\left( U^n ( \hat{m}_0 ), V^n_0 ( \hat{m}_0 , \hat{m}_1 , \hat{m}_r )\right. \hspace{1cm}\\ \hspace{1cm}\left.,V^n_1 ( \hat{m}_0 , \hat{m}_1 , \hat{m}_r , \hat{t}_1 ), Y^n_1 \right) \! \in\!   A^n_\epsilon \\ {\text{for some }(\hat{m}_0,\hat{m}_1,\hat{m}_r) \!=\! (1,1,1),\ \hat{t}_1 \!\neq\!  1}\end{array} \!   \right| \mathcal{I}\! =\! \mathbf{1}\! \right)\nonumber\\\label{forlastapp}\\
 &  \leq   2^{n(R_0+R_1+T_1+R_r-I(UV_0V_1;Y_1)\nonumber+\gamma_2(\epsilon)+\delta(\epsilon))}\nonumber\\&\quad+2^{n(R_1+T_1+R_r-I(V_0V_1;Y_1|U)+\gamma_2^\prime(\epsilon)+\delta(\epsilon))}\label{jrnl1-bound2ex2}
\end{align}}}
In the last step, the first probability term is bounded by $2^{n(R_0+R_1+T_1+R_r-I(UV_0V_1;Y_1)+\gamma_2(\epsilon)+\delta(\epsilon))}$ based on the derivation in Section \ref{onlinedecoder}, and the second probability term is bounded by  $2^{n(R_1+T_1+R_r-I(V_0V_1;Y_1|U)+\gamma_2^\prime(\epsilon)+\delta(\epsilon))}$ for similar reasons (in the conditional form), the details of which are presented in Appendix \ref{app-last-example3}.

It becomes clear from
\eqref{jrnl1-boundauxiliary}, that the auxiliary decoder also succeeds with high
probability under the non-unique decoding constraints of
\cite{ChiaElGamal11}. Similar to Subsection \ref{onlinedecoder}, one can
conclude that if for an operating point the non-unique decoder succeeds with
high probability, then there also exists a joint unique decoding scheme that
succeeds with high probability. 

One can also use the auxiliary decoder to (explicitly) devise the joint unique decoding scheme. Analogous to Subsection \ref{onlinedecoder}, the decoding scheme could be joint unique decoding of $m_0,m_1,m_r$ in the regime 
where it succeeds (with high probability) and joint unique decoding of $m_0,m_1,m_r,t_1$ otherwise. 
To express the two regimes, we analyze the error probability of the component (joint unique) decoder that decodes $m_0$, $m_1$ and $m_r$. 
\begin{align*}
\Pr(\text{error})\leq&\, \epsilon+ 2^{n( R_0+R_1+R_r-I( UV_0;Y_1 )+\sigma(\epsilon) )} \\& + 2^{n( R_1+R_r-I( V_0;Y_1|U )+\sigma(\epsilon) )}
\end{align*} 
where $\sigma(\epsilon)\to 0$ if $\epsilon\to 0$. Therefore, joint unique 
decoding of $m_0$, $m_1$ and $m_r$ succeeds with high probability if the
following two inequalities hold 
in addition to the indirect
decoding constraints of~\cite{ChiaElGamal11}.
\begin{eqnarray}
&&R_0+R_1+R_r<I(UV_0;Y_1)\label{jrnl1-r1}\\
&&R_1+R_r<I(V_0;Y_1|U)\label{jrnl1-r2}
\end{eqnarray}
If either of the above inequalities does not hold, then joint unique decoding of $m_0$, $m_1$, $m_r$ fails with high probability (see Appendix \ref{jrnl1-appendixexample1etc}).  
Nonetheless, while the non-unique decoding constraint
of \cite{ChiaElGamal11} is satisfied, since the auxiliary decoder succeeds with high probability, we conclude that joint unique decoding of $m_0$ ,$m_1$ , $m_r$, $t_1$ succeeds with high probability.
So the following joint unique decoding scheme achieves the inner-bound of \cite{ChiaElGamal11}: If inequalities \eqref{jrnl1-r1} and \eqref{jrnl1-r2} hold, jointly uniquely decode indices $m_0$, $m_1$, and $m_r$, and 
otherwise, jointly uniquely decode all four indices $m_0$, $m_1$, $m_r$, $t_1$.}}

%% file: Appendix.tex
\appendices
\section{ For any $\delta>0$ and jointly typical triples $(u^n,v_2^n,v_3^n)$  $p_{\hat{V}^n_2|U^nV_2^nV_3^n\mathcal{I}}(\hat{v}^n_2|u^n,v_2^n,v_3^n,\mathbf{1})\leq 2^{n\delta}p_{\hat{V}^n_2|U^n}(\hat{v}^n_2|u^n)$ for $n$ large enough} 
\label{app:typicalitylemma}
We proceed along the  lines of \cite[Lemma 1]{MineroLimKim13}. 
Recall the codebook structure, where (i) $V_2^n(m_0,s_0,t_2)$ and $V^n_3(m_0,s_0,t_3)$ are superposed on $U^n(m_0,s_0)$, (ii) $V_2^n(m_0,s_0,t_2)$ and $V_3(m_0,s_0,t_3)$ are distributed into bins $\mathcal{B}_2(m_0,s_0,s_2)$ and $\mathcal{B}_3(m_0,s_0,s_3)$ and (iii) that a jointly typical pair $(V_2^n(m_0,s_0,t_2),V_3^n(m,s_0,t_3))$ is chosen randomly in each product bin. In the error analysis of Section \ref{onlinedecoder}, we assumed all sent indices to be $1$, and we considered the event of decoding a wrong index $\hat{t}_2$ (and thus an incorrect sequence $V_2^n(m_0,s_0,\hat{t}_2)$). 
We denote $U^n(1,1), V^n_2(1,1,1)$, $V^n_3(1,1,1)$, ${V}^n_2(1,1,\hat{t}_2)$ by $U^n$, $V_2^n$, $V_3^n$, $\hat{V}_2^n$, respectively. If $V^n_2(1,1,1)$ and $V_2^n(1,1,\hat{t}_2)$ belong to two different bins $\mathcal{B}_1(1,1,s_2)$ and $\mathcal{B}_1(1,1,s^\prime_2)$, $s_2\neq s^\prime_2$, then it is easy to see that the relation $p_{\hat{V}^n_2|U^nV_2^nV_3^n\mathcal{I}}(\hat{v}^n_2|u^n,v_2^n,v_3^n,\mathbf{1})=p_{\hat{V}^n_2|U^n}(\hat{v}^n_2|u^n)$ holds. Here we only need to consider the case where $\hat{t}_2$ is such that $V^n_2(1,1,1)$ and $V_2^n(1,1,\hat{t}_2)$ belong to the same bin, i.e., $\mathcal{B}_2(1,1,1)$. We assume without any loss of generality that $\hat{t}_2=2$.

Define the random ensemble ${\mathbf{C}}^\prime\in\mathcal{C}^\prime$ as the overall collection of all  sequences $(V_2^n(1,1,t_2))$ and $(V_3^n(1,1,t_3))$, where ${t_2}\in\{3,\ldots,2^{n(T_2-S_2)}\}$ and ${t_3}\in\{2,\ldots,2^{n(T_3-S_3)}\}$. 
For a given $\mathbf{c}^\prime$, define $N_1(v_2^n,v_3^n,\mathbf{c}^\prime)$ to be the number of jointly typical pairs $\left(v_2^n(1,1,t_2),v_3^n(1,1,t_3)\right)$ for all $t_2\neq 2, t_3$. Similarly, given $\mathbf{c}^\prime$ and ${v}_2^n(1,1,2)$, let $N_2(v_2^n(1,1,2),v_3^n,\mathbf{c}^\prime)$  be the number of jointly typical pairs $\left(v_2^n(1,1,2),v_3^n(1,1,t_3)\right)$ for all $t_3$.

 We now write
 \begin{align}
 &p_{\hat{V}_2^n|U^nV_2^nV_3^n\mathcal{I}}(\hat{v}_2^n|u^n,v_2^n,v_3^n,\mathbf{1})\nonumber\\
 &=p_{\hat{V}_2^n|U^nV_2^nV_3^n\mathcal{I}_{t_2}\mathcal{I}_{t_3}}(\hat{v}_2^n|u^n,v_2^n,v_3^n,1,1)\nonumber\\
 &=\sum_{\mathbf{c}^\prime\in\mathcal{{C}}^\prime_{}} p_{\hat{V}_2^n\mathbf{C}^\prime|U^nV_2^nV_3^n\mathcal{I}_{t_2},\mathcal{I}_{t_3}}(\hat{v}_2^n,\mathbf{c}^\prime|u^n,v_2^n,v_3^n,1,1)\nonumber\\
  &= \sum_{\mathbf{c}^\prime\in{\mathcal{C}}^\prime_{}} \!\!\left[p( \mathbf{c}^\prime|u^n , v_2^n, v_3^n, 1, 1 )p_{\hat{V}^n_2|U^n}\!(\hat{v}_{2}^n|u^n , v_2^n, v_3^n, \mathbf{c}^\prime )\vphantom{\times\frac{p(\mathcal{I}_{t_2}=1,\mathcal{I}_{t_3}=1|u^n, v_2^n, v_3^n, \hat{v}^n_2, \mathbf{c}^\prime )}{p(\mathcal{I}_{t_2}=1,\mathcal{I}_{t_3}=1|u^n,v_2^n,v_3^n,\mathbf{c}^\prime)}}\right.\nonumber\\&\hspace{1.1cm}\left.\times\frac{p(\mathcal{I}_{t_2}\!\!\!=\!\!1,\mathcal{I}_{t_3}\!\!=\!\!1|u^n, v_2^n, v_3^n, \hat{v}^n_2, \mathbf{c}^\prime )}{p(\mathcal{I}_{t_2}\!\!=\!\!1,\mathcal{I}_{t_3}\!\!=\!\!1|u^n,v_2^n,v_3^n,\mathbf{c}^\prime)}\right]\nonumber\\
   &=p_{\hat{V}^n_2|U^n}(\hat{v}_{2}^n|u^n)\!\sum_{\mathbf{c}^\prime\in{\mathcal{C}}^\prime_{}} \!\!\left[p( \mathbf{c}^\prime|u^n , v_2^n, v_3^n, 1, 1 )\vphantom{\times\frac{p(\mathcal{I}_{t_2}\!\!=\!\!1,\mathcal{I}_{t_3}\!=\!1|u^n, v_2^n, v_3^n, \hat{v}^n_2, \mathbf{c}^\prime )}{p(\mathcal{I}_{t_2}\!=\!1,\mathcal{I}_{t_3}\!=\!1|u^n,v_2^n,v_3^n,\mathbf{c}^\prime)}}\right.\nonumber\\&\hspace{3.25cm}\left.\times\frac{p(\mathcal{I}_{t_2}\!\!=\!\!1,\mathcal{I}_{t_3}\!\!=\!\!1|u^n\!, v_2^n, v_3^n, \hat{v}^n_2, \mathbf{c}^\prime )}{p(\mathcal{I}_{t_2}\!\!=\!\!1,\mathcal{I}_{t_3}\!\!=\!\!1|u^n\!,v_2^n,v_3^n,\mathbf{c}^\prime)}\!\right]\label{tocontinue}
  \end{align}

To continue bounding \eqref{tocontinue}, we consider two cases.
 \begin{enumerate}
 \item $T_3-S_3-I(V_2;V_3|U)<0$: We bound the fraction in \eqref{tocontinue}. The numerator is bounded from above by disregarding $\hat{v}_2^n$.
\begin{eqnarray}
 &&p(\mathcal{I}_{t_2}=1,\mathcal{I}_{t_3}=1|u^n,v_2^n,v_3^n,\hat{v}_2^n,\mathbf{c}^\prime)\nonumber\\
 &&=\frac{1}{N_1(v_2^n,v_3^n,\mathbf{c}^\prime)+N_2(\hat{v}_2^n,v_3^n,\mathbf{c}^\prime)}\nonumber\\
&&\leq \frac{1}{N_1(v_2^n,v_3^n,\mathbf{c}^\prime)}.\label{likebefore1}
 \end{eqnarray}
 
 The denominator is bounded from below by the expression in \eqref{likebefore2}.
\begin{align}
 &p(\mathcal{I}_{t_2}=1,\mathcal{I}_{t_3}=1|u^n,v_2^n,v_3^n,\mathbf{c}^\prime)\nonumber\\
 &\geq p(\mathcal{I}_{t_2}\!\!=\!\!1,\mathcal{I}_{t_3}\!\!=\!\!1, N_2(\hat{V}_2^n,v_3^n,\mathbf{c}^\prime)\!\!=\!\!0|u^n , v_2^n, v_3^n, \mathbf{c}^\prime )\nonumber\\
 &\geq p(N_2(\hat{V}_2^n,v_3^n,\mathbf{c}^\prime)=0|u^n,v_2^n,v_3^n,\mathbf{c}^\prime)\nonumber\\
 &\hspace{.35cm}\times p(\mathcal{I}_{t_2}\!\!=\!\!1,\mathcal{I}_{t_3}\!\!=\!\!1| N_2( \hat{V}_2^n, v_3^n, \mathbf{c}^\prime )\!\!=\!\!0, u^n , v_2^n, v_3^n, \mathbf{c}^\prime )\nonumber\\
  &= p(N_2(\hat{V}_2^n,v_3^n,\mathbf{c}^\prime)=0|u^n,v_2^n,v_3^n,\mathbf{c}^\prime)\nonumber\\
  &\hspace{.35cm}\times\frac{1}{N_1(v_2^n,v_3^n,\mathbf{c}^\prime)}\nonumber\\
  &\geq \left(1-2^{n(T_3-S_3)}2^{-n(I(V_2;V_3|U)-\sigma(\epsilon))}\right)\nonumber\\
  &\hspace{.35cm} \times \frac{1}{N_1(v_2^n,v_3^n,\mathbf{c}^\prime)}\label{likebefore2}
 \end{align}
 Comparing \eqref{likebefore1} and \eqref{likebefore2} (under the assumption  that {$T_3-S_3-I(V_2;V_3|U)<0$}), it becomes clear that we have  $p(\hat{v}_2^n|u^n,v_2^n,v_3^n)\leq 2^{n\delta}p(\hat{v}_2^n|u^n)$ for every $\delta>0$ and $n$ large enough.
 \item  $T_3-S_3-I(V_2;V_3|U)> 0$: in this case, we first re-write  expression  \eqref{tocontinue} as follows.

\begin{align}
&p_{\hat{V}^n_2|U^n}\!(\hat{v}_{2}^n|u^n)\!\! \sum_{\mathbf{c}^\prime\in{\mathcal{C}}^\prime_{}}\!\! \left[p(\mathbf{c}^\prime|u^n\!,v_2^n,v_3^n,1,1)\vphantom{\times\frac{p(\mathcal{I}_{t_2}=1,\mathcal{I}_{t_3}=1|u^n\!,v_2^n,v_3^n,\hat{v}^n_2,\mathbf{c}^\prime)}{p(\mathcal{I}_{t_2}\!\!=\!\!1,\mathcal{I}_{t_3}\!\!=\!\!1|u^n\!,v_2^n,v_3^n,\mathbf{c}^\prime)}}\right.\nonumber\\&\hspace{2.7cm}\left.\times\frac{p(\mathcal{I}_{t_2}\!\!=\!\!1,\mathcal{I}_{t_3}\!\!=\!\!1|u^n\!,v_2^n,v_3^n,\hat{v}^n_2,\mathbf{c}^\prime)}{p(\mathcal{I}_{t_2}\!\!=\!\!1,\mathcal{I}_{t_3}\!\!=\!\!1|u^n\!,v_2^n,v_3^n,\mathbf{c}^\prime)}\!\right]\nonumber\\ 
&=p(\hat{v}_{2}^n|u^n)\!\!\sum_{\mathbf{c}^\prime\in{\mathcal{C}}^\prime_{}}\!\!\left[ \frac{p(\mathcal{I}_{t_2}\!\!=\!\!1,\mathcal{I}_{t_3}\!\!=\!\!1,\mathbf{c}^\prime|u^n\!,v_2^n,v_3^n)}{p(\mathcal{I}_{t_2}\!\!=\!\!1,\mathcal{I}_{t_3}\!\!=\!\!1|u^n\!,v_2^n,v_3^n)}\right.\nonumber\\&\hspace{2.3cm}\left.\vphantom{}\times\frac{p(\mathcal{I}_{t_2}\!\!=\!\!1,\mathcal{I}_{t_3}\!\!=\!\!1|u^n\!,v_2^n,v_3^n,\hat{v}^n_2,\mathbf{c}^\prime)}{p(\mathcal{I}_{t_2}\!\!=\!\!1,\mathcal{I}_{t_3}\!\!=\!\!1|u^n\!,v_2^n,v_3^n,\mathbf{c}^\prime)}\right]\nonumber\\
   &=p(\hat{v}_{2}^n|u^n)\!\!\sum_{\mathbf{c}^\prime\in{\mathcal{C}}^\prime_{}}\!\! \left[\frac{p(\mathbf{c}^\prime|u^n\!,v_2^n,v_3^n)}{p(\mathcal{I}_{t_2}\!\!=\!\!1,\mathcal{I}_{t_3}\!\!=\!\!1|u^n\!,v_2^n,v_3^n)}\right.\nonumber\\&\hspace{2.3cm}\left.{ \vphantom{\frac{p(\mathbf{c}^\prime|u^n\!,v_2^n,v_3^n)}{p(\mathcal{I}_{t_2}\!\!=\!\!1,\mathcal{I}_{t_3}\!\!=\!\!1|u^n\!,v_2^n,v_3^n)}}\times p(\mathcal{I}_{t_2}\!\!=\!\!1,\mathcal{I}_{t_3}\!\!=\!\!1|u^n\!,v_2^n,v_3^n,\hat{v}^n_2,\mathbf{c}^\prime)}\right]\nonumber\\
      &=p(\hat{v}_{2}^n|u^n)\!\!\sum_{\mathbf{c}^\prime\in{\mathcal{C}}^\prime_{}}\!\! \left[\frac{p(\mathbf{c}^\prime|u^n\!,v_2^n,v_3^n,\hat{v}_2^n)}{p(\mathcal{I}_{t_2}\!\!=\!\!1,\mathcal{I}_{t_3}\!\!=\!\!1|u^n\!,v_2^n,v_3^n)}\right.\nonumber\\&\hspace{2.3cm}\left.\vphantom{\frac{p(\mathbf{c}^\prime|u^n\!,v_2^n,v_3^n,\hat{v}_2^n)}{p(\mathcal{I}_{t_2}\!\!=\!\!1,\mathcal{I}_{t_3}\!\!=\!\!1|u^n\!,v_2^n,v_3^n)}} \times p(\mathcal{I}_{t_2}\!\!=\!\!1,\mathcal{I}_{t_3}\!\!=\!\!1|u^n\!,v_2^n,v_3^n,\hat{v}^n_2,\mathbf{c}^\prime)\right]\nonumber\\
  &= p(\hat{v}_{2}^n|u^n)\frac{p(\mathcal{I}_{t_2}\!=\!1,\mathcal{I}_{t_3}\!=\!1|u^n\!,v_2^n,v_3^n,\hat{v}^n_2)}{p(\mathcal{I}_{t_2}\!=\!1,\mathcal{I}_{t_3}\!=\!1|u^n\!,v_2^n,v_3^n)}\label{fraction}
%
 \end{align}

The following claim will be the key in bounding the fraction in \eqref{fraction}. 
 \begin{claim}
 \label{dblexp}
Let all sequences $V_2^n(1,1,t_2)$, $t_2\neq1$, and $V_3^n(1,1,t_3)$, $t_3\neq 1$, be picked randomly in the product bin of interest. The event where the number of jointly typical pairs in a row is much larger than the total remaining number of jointly typical pairs in the bin has a probability which decays to zero double exponentially fast with $n$ (under the assumption $T_3-S_3>I(V_2;V_3|U)$); i.e., for some constant $\alpha,\beta>0$,
 \begin{align}
&\Pr\!\left(\!\!\!\begin{array}{l|l}N_2( \hat{V}_2^n, v_3^n, \mathbf{C}^\prime )\!>\!2^{2+2n\delta(\epsilon)}\!N_1( v_2^n, v_3^n, \mathbf{C}^\prime )\!\!&\!U^n\!\!=\!u^n\end{array}\!\!\!\!\right)\nonumber\\&\leq \beta\exp\left(\!-\alpha2^{ n(T_3-S_3-I(V_2;V_3|U)-\delta(\epsilon))}\!\right)\!
 \end{align}
 \end{claim}
 \begin{proof}
 Let $N_3(v_3^n,\mathbf{C}^\prime)$ be the number of jointly typical pairs $(V_2^n(1,1,3),V_3^n(1,1,t_3))$, where $t_3=1,\ldots,2^{n(T_3-S_3)}$. Obviously, $N_1(v_2^n,v_3^n,\mathbf{C}^\prime)\geq N_3(v_3^n,\mathbf{C}^\prime)$. To prove the claim, it is sufficient to show that
 \begin{align}
 &\Pr \!\left(\! N_2( \hat{V}_2^n\!\hspace{-.025cm}, \hspace{-.025cm}v_3^n\! , \hspace{-.025cm}\mathbf{C}^\prime)\!\hspace{-.05cm} >\! \hspace{-.05cm} 2^{1 + n\delta(\hspace{-.05cm}\epsilon\hspace{-.05cm})} 2^{n( T_3 \hspace{-.025cm}-\hspace{-.025cm} S_3\hspace{-.025cm}-\hspace{-.025cm}I(\! V_2;V_3|U\! )\!)} \!\right|\!\!\left.\vphantom{N_2( \hat{V}_2^n , v_3^n , \mathbf{C}^\prime )\! >\!  2^{1 + n\delta(\epsilon)} 2^{n( T_3 - S_3 -I(\! V_2;V_3|U\! ))} }U^n\!\!=\! u^n\!\!\right)\nonumber\\
 &\leq \beta_1\exp\!\left(\! -\alpha_12^{ n(T_3-S_3-I(V_2;V_3|U)-\delta(\epsilon)\!)}\! \right)\label{lastineq1}\end{align}
for some $\alpha_1,\beta_1>0$, and that
  \begin{align}
  &\Pr\! \left(\! N_3 ( v_3^n , \mathbf{C}^\prime )\! <\!  2^{-1 - n\delta(\hspace{-.05cm}\epsilon\hspace{-.05cm})} 2^{n( T_3 - S_3 -I(\! V_2;V_3|U\!)\! )} \!\right|\!\!\left.\vphantom{N_3 ( v_3^n , \mathbf{C}^\prime )\! <\!  2^{-1 - n\delta(\epsilon)} 2^{n( T_3 - S_3 -I(\! V_2;V_3|U \!))}}U^n\!\!\!=\! u^n\!\!\right)\nonumber\\
  &\leq \beta_2\exp\!\left(\! -\alpha_22^{ n(T_3-S_3-I(V_2;V_3|U)-\delta(\epsilon)\!)}\! \right)\label{lastineq2}
  \end{align}
for some $\alpha_2,\beta_2>0$.
Both of the above  inequalities can be shown using standard Chernoff bounding techniques. We defer the interested reader to Appendix \ref{LASTAPP}.
 \end{proof}
The numerator of \eqref{fraction} is bounded from above by disregarding $\hat{v}_2^n$.
 \begin{align}
\Pr&(\mathcal{I}_{t_2}=1,\mathcal{I}_{t_3}=1|u^n,v_2^n,\hat{v}_2^n,v_3^n)\nonumber\\
&\qquad\qquad\leq \sum_{\mathbf{c}^\prime}p(\mathbf{c}^\prime|u^n)\frac{1}{N_1(v_2^n,v_3^n,\mathbf{c}^\prime)}\nonumber\\
&\qquad\qquad= \mathbb{E}\left[\left.\frac{1}{N_1(v_2^n,v_3^n,\mathbf{C}^\prime)}\right|U^n=u^n\right]\label{final1num}
 \end{align}
 Similarly for the denominator we have the lower bound in \eqref{final1denom}. 
 \allowdisplaybreaks
 \begin{align}
&\Pr(\mathcal{I}_{t_2}=1,\mathcal{I}_{t_3}=1|u^n,v_2^n,v_3^n)\nonumber\\
&=\sum_{\mathbf{c}^\prime,\tilde{v}_2^n}\left[p_{\mathbf{C}^\prime\hat{V}_2^n|U^n}(\mathbf{c}^\prime,\tilde{v}_2^n|u^n)\right.\nonumber\\&\hspace{1.25cm}\left.\times p_{\mathcal{I}_{t_2}\mathcal{I}_{t_3}|U^nV_2^nV_3^n\hat{V}_2^n\mathbf{C}^\prime}(1,1|u^n\!,v_2^n,v_3^n,\tilde{v}_2^n,\mathbf{c}^\prime)\right]\nonumber\\
&=\sum_{\substack{\mathbf{c}^\prime,\tilde{v}_2^n}}\left[p_{\mathbf{C}^\prime\hat{V}_2^n|U^n}(\mathbf{c}^\prime,\tilde{v}_2^n|u^n)\vphantom{\frac{1}{N_1(v_2^n,v_3^n,\mathbf{c}^\prime)+N_2(\tilde{v}_2^n,v_3^n,\mathbf{c}^\prime)}}\right.\nonumber\\&\hspace{1.25cm}\left.\times \frac{1}{N_1(v_2^n,v_3^n,\mathbf{c}^\prime)+N_2(\tilde{v}_2^n,v_3^n,\mathbf{c}^\prime)}\right]\nonumber\\
&\geq\hspace{-3.9cm}\sum_{\substack{\mathbf{c}^\prime,\tilde{v}_2^n:\\\hspace{3.9cm}N_2(\tilde{v}_2^n,v_3^n,\mathbf{c}^\prime)\leq 2^{2+2n\delta(\epsilon)}N_1(v_2^n,v_3^n,\mathbf{c}^\prime)}}\hspace{-3.9cm}\left[p_{\mathbf{C}^\prime\hat{V}_2^n|U^n}(\mathbf{c}^\prime,\tilde{v}_2^n|u^n)\vphantom{\frac{1}{N_1(v_2^n,v_3^n,\mathbf{c}^\prime)+N_2(\tilde{v}_2^n,v_3^n,\mathbf{c}^\prime)}}\right.\nonumber\\&\left.\hspace{1.25cm}\times\frac{1}{N_1(v_2^n,v_3^n,\mathbf{c}^\prime)+N_2(\tilde{v}_2^n,v_3^n,\mathbf{c}^\prime)}\right]\nonumber\\
&\geq\hspace{-3.9cm}\sum_{\substack{\mathbf{c}^\prime,\tilde{v}_2^n:\\\hspace{3.9cm}N_2(\tilde{v}_2^n,v_3^n,\mathbf{c}^\prime)\leq 2^{2+2n\delta(\epsilon)}N_1(v_2^n,v_3^n,\mathbf{c}^\prime)}}\hspace{-3.9cm}\left[p_{\mathbf{C}^\prime\hat{V}_2^n|U^n}(\hspace{-.025cm}\mathbf{c}^\prime\!,\tilde{v}_2^n|u^n\hspace{-.025cm})2^{-3-2n\delta(\epsilon)}\hspace{-.05cm}\frac{1}{N_1(\hspace{-.025cm}v_2^n,v_3^n,\mathbf{c}^\prime\hspace{-.025cm})}\!\right]\nonumber\\
&=2^{-3-2n\delta(\hspace{-.025cm}\epsilon\hspace{-.025cm})}\!\!\!\left[\!\sum_{\substack{\mathbf{c}^\prime,\tilde{v}_2^n}}p_{\mathbf{C}^\prime\hat{V}_2^n|U^n}\!(\mathbf{c}^\prime\!,\tilde{v}_2^n|u^n)\frac{1}{N_1(v_2^n,v_3^n,\mathbf{c}^\prime)}\vphantom{\sum_{\substack{\mathbf{c}^\prime,\tilde{v}_2^n:\\\hspace{3.9cm}N_2(\tilde{v}_2^n,v_3^n,\mathbf{c}^\prime)>2^{2+2n\delta(\epsilon)}N_1(v_2^n,v_3^n,\mathbf{c}^\prime)}}}\right.\nonumber\\&\hspace{2cm}-\!\!\!\!\left.\hspace{-3.8cm}\sum_{\substack{\mathbf{c}^\prime\!,\tilde{v}_2^n:\\\hspace{3.9cm}N_2(\tilde{v}_2^n,v_3^n,\mathbf{c}^\prime)>2^{2+2n\delta(\epsilon)}N_1(v_2^n,v_3^n,\mathbf{c}^\prime)}}\hspace{-2.1cm}\hspace{-1.9cm}p_{\mathbf{C}^\prime\hat{V}_2^n|U^n}\!(\mathbf{c}^\prime\!,\tilde{v}_2^n|u^n)\frac{1}{N_1(v_2^n\!,v_3^n\!,\mathbf{c}^\prime)}\!\right]\nonumber\\
&\stackrel{(a)}{\geq}2^{-3-2n\delta(\epsilon)}\mathbb{E}\left[\begin{array}{l|l}\frac{1}{N_1(v_2^n,v_3^n,\mathbf{C}^\prime)}&U^n=u^n\end{array}\right]\nonumber\\&\quad-2^{-3-2n\delta(\epsilon)}\Pr\!\left(\!\!\!\!\!\!\!\begin{array}{l|l}\begin{array}{l}N_2(\hat{V}_2^n\!,v_3^n\!,\mathbf{C}^\prime)>\\\hspace{.25cm}2^{2+2n\delta(\epsilon)}N_1(v_2^n\!,v_3^n\!,\mathbf{C}^\prime)\end{array}\!\!\!\!\!\!&\!\!U^n\!=\!u^n\end{array}\hspace{-.3cm}\right)\nonumber\\
&\stackrel{(b)}{\geq}2^{-3-2n\delta(\epsilon)}(1-c)\mathbb{E}\left[\!\!\!\!\begin{array}{l|l}\frac{1}{N_1(v_2^n\!,v_3^n\!,\mathbf{C}^\prime)}\!\!&\!\!U^n\!=\!u^n\end{array}\!\!\!\!\right],\quad c\!>\!0\label{final1denom}
 \end{align}

In the above, (a) holds because $N_1(v_2^n,v_3^n,\mathbf{C}^\prime)\geq 1$ (ensured by the assumption that $(v_2^n,v_3^n)\in\mathcal{A}_\epsilon^n$). Step (b) holds, for any constant $c>0$ and large enough $n$, by Claim \ref{dblexp} as follows.
\allowdisplaybreaks
\begin{align*}
&\Pr(N_2(\hat{V}_2^n\!,v_3^n\!,\mathbf{C}^\prime)\!>\!2^{2+2n\delta(\epsilon)}N_1(v_2^n\!,v_3^n\!,\mathbf{C}^\prime)|U^n\!=\!u^n)\\
&\leq \beta\exp\left(-\alpha2^{ n(T_3-S_3-I(V_2;V_3|U)-\delta(\epsilon))}\right)\\
&\stackrel{(a)}{\leq} \frac{c}{2}2^{-n(T_2-S_2+T_3-S_3-I(V_2;V_3|U)+\delta(\epsilon))}\\
&\leq c\frac{1}{\mathbb{E}\left[N_1(v_2^n,v_3^n,\mathbf{C}^\prime)|U^n=u^n\right]}\\
&\leq c\mathbb{E}\left[\begin{array}{l|l}\frac{1}{N_1(v_2^n,v_3^n,\mathbf{C}^\prime)}&U^n=u^n\end{array}\right]
\end{align*}
In step (a) above, we have used the fact that $T_3-S_3>I(V_2;V_3|U)-\delta(\epsilon)$, $T_2\geq S_2$, and that $n$ is large enough.

Finally, upper bounding the numerator of \eqref{fraction} by \eqref{final1num} and lower bounding its denominator by \eqref{final1denom}, we reach to a factor with an exponent of order ${n\delta(\epsilon)}$. Inserting this back into \eqref{fraction}, we conclude that for every $\delta>0$ and $n$ large enough $p(\hat{v}_2^n|u^n,v_2^n,v_3^n)\leq 2^{n\delta}p(\hat{v}_2^n|u^n)$.

\end{enumerate}

\section{Proof to Lemma \ref{jrnl1-Lemmaexistence}}
\label{app:Lemmaexistence}
We start by proving the following claim. 
\begin{claim}
 \label{jrnl1-claimconverse}
 Component decoder $1$ succeeds with high probability (averaged over
codebooks) if $R_0+S_0< I(U;Y_2)$, and fails with high probability, if
$R_0+S_0>I(U;Y_2)$.
 \end{claim}
\begin{proof}[Proof of Claim \ref{jrnl1-claimconverse}]
Component decoder $1$ makes an error in decoding only if one of the following events occur:
\begin{enumerate}
\item[({\em i})] $(U^n(1,1),Y^n_2)$ is not jointly typical. The probability of this event can be made arbitrarily small by choosing a large enough $n$.
\item[({\em ii})] There exists a pair of indices $(\hat{m}_0,\hat{s}_0)\neq(1,1)$ such that $(U^n(\hat{m}_0,\hat{s}_0),Y^n_2)$ is jointly typical.
\end{enumerate}
To analyze the error probability, we assume without any loss of generality that the originally sent indices are $m_0=1$ and $s_0=1$. The error probability is thus upper-bounded by
\begin{eqnarray*}
&&\hspace{-.7cm}\Pr(\text{error at component decoder }1)\nonumber\\
&\leq&    \epsilon + \Pr\left(\!\!\!\!\!\!\left.\begin{array}{c}( U^n (\hat{m}_0,\hat{s}_0),Y^n_2 ) \in  \mathcal{A}^n_\epsilon\\\text{ for some }  (\hat{m}_0,\hat{s}_0)  \neq   (1,1)\end{array}\right|m_0=1, s_0=1\!\!\right)\nonumber\\
&\leq&    \epsilon + 2^{n(R_0+S_0-I(U;Y_2)+\delta(\epsilon))},
\end{eqnarray*}
where $\delta(\epsilon)\to 0$ if $\epsilon\to 0$. This proves that for large enough $n$, the error probability of component decoder $1$ could be made arbitrary small if $R_0+S_0< I(U;Y_2)$.

On the other hand, decoder $1$ makes an error if there exists an index pair $(\hat{m}_0,\hat{s}_0)\neq (1,1)$ such that $(U^n(\hat{m}_0,\hat{s}_0),Y^n_2)$ is jointly typical. The probability of error at decoder $1$ is, therefore, lower-bounded by 
\begin{eqnarray}
\Pr\left(\begin{array}{l|l}\!\!\!\!\begin{array}{c}(U^n(\hat{m}_0,\hat{s}_0),Y^n_2)\in \mathcal{A}^n_\epsilon\\\text{ for some } (\hat{m}_0,\hat{s}_0)\neq (1,1)\end{array}&m_0=1, s_0=1\end{array}\right),
\end{eqnarray} and we want to show that it is arbitrarily close to $1$ if $R_0+S_0>I(U;Y_2)$. We instead look at the complementary event, $\{(U^n(\hat{m}_0,\hat{s}_0),Y^n_2)\notin \mathcal{A}^n_\epsilon\text{ for all } (\hat{m}_0,\hat{s}_0)\neq (1,1)\}$, and show that its probability can be made arbitrarily small.
\allowdisplaybreaks
\begin{align}
&\Pr\left(\!\!\!\!\begin{array}{l|l}\begin{array}{c}(U^n(\hat{m}_0,\hat{s}_0),Y^n_2) \!\notin\!  \mathcal{A}^n_\epsilon\\\text{ for all } (\hat{m}_0,\hat{s}_0) \!\neq\!  (1,1)\end{array}&m_0\!=\!1, s_0\!=\!1\end{array}\!\!\!\!\right)\nonumber\\
  &=   \sum_{y_2^n}\left[\Pr(Y^n_2 = y_2^n| m_0\!=\!1, s_0\!=\!1)\vphantom{\times\Pr \left( \!\!\!\!\!\!\!\!\begin{array}{c|c}\begin{array}{c}(U^n(\hat{m}_0,\hat{s}_0),y_2^n) \notin  \mathcal{A}^n_\epsilon\\\text{ for all } (\hat{m}_0,\hat{s}_0) \neq  (1,1)\end{array}    &    \begin{array}{l}Y^n_2\! =\! y_2^n,\\ m_0\!=\!1,\\ s_0\!=\!1\end{array}\end{array}   \!\!\!\!\!\!\!\!  \right)}\right.\nonumber\\&\hspace{1.1cm}\left.\times\Pr \left( \!\!\!\!\!\!\!\!\begin{array}{c|c}\begin{array}{c}(U^n(\hat{m}_0,\hat{s}_0),y_2^n) \notin  \mathcal{A}^n_\epsilon\\\text{ for all } (\hat{m}_0,\hat{s}_0) \neq  (1,1)\end{array}     \!\!\!\!  & \!\!    \begin{array}{l}Y^n_2\! =\! y_2^n,\\ m_0\!=\!1,\\ s_0\!=\!1\end{array}\end{array}   \!\!\!\!\!\!  \right)\right]\nonumber\\
&\leq  \epsilon+   \sum_{y_2^n \in \mathcal{A}^n_\epsilon} \left[  \Pr(Y^n_2 = y_2^n| m_0\!=\!1, s_0\!=\!1)\vphantom{\times\Pr \left(\!\!\!\!\!\!\!\!\!\!\begin{array}{c|c}\begin{array}{c}(U^n(\hat{m}_0,\hat{s}_0),y_2^n) \notin  \mathcal{A}^n_\epsilon\\\text{ for all } (\hat{m}_0,\hat{s}_0) \neq  (1,1)\end{array}  \!\!\!\!  & \!\!  \begin{array}{l}Y^n_2\! =\! y_2^n,\\ m_0\!=\!1,\\ s_0\!=\!1\end{array}\end{array}  \!\!\!\! \!\!\!\!  \right)}\right.\nonumber\\&\hspace{2.1cm}\left. \times\Pr \left(\!\!\!\!\!\!\!\!\begin{array}{c|c}\begin{array}{c}(U^n(\hat{m}_0,\hat{s}_0),y_2^n) \notin  \mathcal{A}^n_\epsilon\\\text{ for all } (\hat{m}_0,\hat{s}_0) \neq  (1,1)\end{array}  \!\!\!\!  & \!\!  \begin{array}{l}Y^n_2\! =\! y_2^n,\\ m_0\!=\!1,\\ s_0\!=\!1\end{array}\end{array}  \!\!\!\! \!\!  \right)\!\right]\nonumber\\
& =  \epsilon+ \sum_{y_2^n \in \mathcal{A}^n_\epsilon} \!\left[  \Pr(Y^n_2 = y_2^n| m_0\!=\!1, s_0\!=\!1) \vphantom{\times\!\prod_{\!\!\!\!\substack{\\\ \\(\hat{m}_0,\hat{s}_0)\neq (1,1)}} \!\! \!\!\!\!\!\!\!   \Pr\left( \begin{array}{c|c}\hspace{-.2cm}(U^n (\hat{m}_0,\hat{s}_0),y_2^n\! ) \!\notin\!  \mathcal{A}^n_\epsilon\!\!\!&\!\!\!\!\! \begin{array}{l}Y^n_2\!\! =\!\! y_2^n,\\ m_0\!=\!1,\\ s_0\!=\!1\end{array}\end{array} \!\!\!\!\!\!\!\right)}\right.\nonumber\\&\hspace{2cm}\left. \times\!\!\!\!\!\prod_{\!\!\!\!\substack{\\\ \\(\hat{m}_0,\hat{s}_0)\neq (1,1)}} \!\! \!\!\!\!\!\!\!   \Pr\left( \begin{array}{c|c}\!\!\!\!\!(U^n (\hat{m}_0,\hat{s}_0),y_2^n\! ) \!\notin\!  \mathcal{A}^n_\epsilon\!\!&\!\!\!\! \begin{array}{l}Y^n_2\!\! =\!\! y_2^n,\\ m_0\!=\!1,\\ s_0\!=\!1\end{array}\end{array} \!\!\!\!\!\!\right)\!\!\right]\nonumber\\
&\leq   \epsilon+ \!\!\!\!\!  \sum_{y_2^n \in \mathcal{A}^n_\epsilon}\left[ \Pr(Y^n_2 = y_2^n|m_0\!=\!1,s_0\!=\!1)\vphantom{\times\left(1 - (1 - \epsilon)2^{-n(I(U;Y_2)+2\epsilon)}\right)^{\left(2^{n(R_0+S_0)}-1\right)}}\right.\nonumber\\&\hspace{1.75cm}\left.\times\left(1 - (1 - \epsilon)2^{-n(I(U;Y_2)+2\epsilon)}\right)^{\left(2^{n(R_0+S_0)}-1\right)}\right]\nonumber\\
&\leq\epsilon+\left(1 - (1 - \epsilon)2^{-n(I(U;Y_2)+2\epsilon)}\right)^{\left(2^{n(R_0+S_0)}-1\right)}\nonumber.
\end{align}
In the limit of $n\rightarrow\infty$, we have
\begin{eqnarray*}
&&\lim_{n\to \infty}\left(1 - (1 - \epsilon)2^{-n(I(U;Y_2)+2\epsilon)}\right)^{\left(2^{n(R_0+S_0)}-1\right)}\\&&=\lim_{n\to \infty}\exp\left\{-\left(2^{n(R_0+S_0)}(1 - \epsilon)2^{-n(I(U;Y_2)+2\epsilon)}\right)\right\},
\end{eqnarray*}
which (for any $0< \epsilon<1$) goes to $0$ as $n$ grows large, if $R_0+S_0>I(U;Y_2)+2\epsilon$.
\end{proof}

From Claim \ref{jrnl1-claimconverse}, it becomes clear that for each
operating point, averaged over codebooks, component decoder~1 either
succeeds with high probability if $R_0+S_0< I(U;Y_2)$ or fails with high
probability if $R_0+S_0> I(U;Y_2)$. 
In the former case, we let the joint unique decoding scheme be that of decoder
$1$, and in the latter, we let the joint unique decoding scheme be that of decoder
$2$. We prove in the following that this joint unique decoding scheme is reliable
(averaged over the codebooks) since the auxiliary decoder is reliable.

Consider an operating point for which decoder $1$ fails with high
probability. In such cases, we assumed the decoding scheme to be
joint unique decoding of messages $M_0$, $M_{10}$, and $M_{12}$. 
For this operating point, the probability of error of our joint unique decoder is 
\begin{eqnarray}
&& \hspace{-.7cm} \Pr( \text{error at component decoder $2$})\nonumber\\
&\leq
    &\Pr\left(\begin{array}{l} \text{error at component decoder $2$ }\\\text{ and component decoder $1$ succeeds}\end{array}\right)\nonumber\\
       && + \Pr\left( \begin{array}{l}\text{error at component decoder $2$}\\ \text{and component decoder $1$ fails}\end{array}\right)\nonumber\\
&\stackrel{(a)}{\leq}& \delta+\Pr\left(\begin{array}{l} \text{error at component decoder $2$ }\\\text{ and component decoder $1$ fails}\end{array}\right)\nonumber\\
&\leq&\delta+\Pr\left(\text{error at the auxiliary decoder}\right).\nonumber
\end{eqnarray}
In the above chain of inequalities, $(a)$ follows from the assumption on
the operating point. Also, $\delta$ and $\epsilon$ can both be taken
arbitrarily close to $0$ for large enough $n$. It is now easy to see that
given an operating point for which component decoder~1 fails, component
decoder~2 succeeds with high probability if the auxiliary decoder succeeds
with high probability.

\section{The error probability analysis of \eqref{jrnl-example2} (Section~\ref{statenoncausally})}
\label{appendixnoncausally}
{We proceed as in Section \ref{onlinedecoder}. We start by splitting the error event into two events:
\allowdisplaybreaks
\begin{align}
&\Pr\!\left( \left.\begin{array}{c} (W^n(1,1),S^n,U^n(1,1,1)\hspace{1.8cm}\\\hspace{1cm},V^n(1,1,1),Y_1^n)\in\mathcal{A}_\epsilon^n,\\\text{ and }\\(W^n(\tilde{m},\tilde{l}_0),Y^n_1)\in A^n_\epsilon \\  \text{for some }(\tilde{m},\tilde{l}_0)\neq(1,1),\\\text{ and } \\(W^n(\hat{m},\hat{l}_0),U^n(\hat{m},\hat{l}_0,\hat{l}_1),Y^n_1)\!\in\! A^n_\epsilon \\ \text{for some }(\hat{m},\hat{l}_0,\hat{l}_1)\neq(1,1,1)\end{array}\right|\mathcal{I}=\mathbf{1} \right)\nonumber\\
&\leq\epsilon+\Pr\!\left(\!\!\left.\begin{array}{c} (W^n(1,1),S^n,U^n(1,1,1)\hspace{1.8cm}\\\hspace{1cm},V^n(1,1,1),Y_1^n)\in\mathcal{A}_\epsilon^n,\\\text{ and }\\(W^n(\tilde{m},\tilde{l}_0),Y^n_1)\in A^n_\epsilon \\ \text{for some }(\tilde{m},\tilde{l}_0)\neq(1,1),\\\text{ and } \\(W^n(\hat{m},\hat{l}_0),U^n(\hat{m},\hat{l}_0,\hat{l}_1),Y^n_1)\!\in\! A^n_\epsilon \\ \text{for some }(\hat{m},\hat{l}_0)\neq(1,1)\text{ and }\ \hat{l}_1\end{array}\!\!\right|\mathcal{I}=\mathbf{1}\! \right)
\nonumber\\
&\quad+\Pr\!\left(\left.\begin{array}{c}(W^n(1,1),S^n,U^n(1,1,1)\hspace{1.8cm}\\\hspace{1cm},V^n(1,1,1),Y_1^n)\in\mathcal{A}_\epsilon^n,\\\text{ and }\\
(W^n(\tilde{m},\tilde{l}_0),Y^n_1)\in A^n_\epsilon \\ \text{for some
}(\tilde{m}_0,\tilde{l}_0)\neq(1,1),\\\text{ and
}\\(W^n(1,1),U^n(1,1,\hat{l}_1),Y^n_1)\!\in\!
A^n_\epsilon  \\\text{ for some } \hat{l}_1\!\neq\! 1\end{array}\hspace{-.1cm}\right|\mathcal{I}=\mathbf{1} \right)\label{term1new}
\end{align}

The first term in \eqref{term1new} is bounded by 
\begin{eqnarray}
\Pr\!\hspace{-.02cm}\left( \begin{array}{l|l}\begin{array}{c}(W^n(\bar{m},\bar{l}_0),U^n(\bar{m},\bar{l}_0,\bar{l}_1),Y^n_1)\!\in\! A^n_\epsilon \\ \text{for some }(\bar{m},\bar{l}_0)\!\neq\!(1,1)\text{ and } \bar{l}_1\end{array} &\mathcal{I}\!=\!\mathbf{1}  \end{array}\right).
\end{eqnarray} If $W^n(\bar{m},\bar{l}_0)$ was independent of $Y^n_1$ (for $(\bar{m},\bar{l}_0)\neq (1,1)$), this would have been the non-unique decoding error probability. However, the conditioning on $\mathcal{I}=\mathbf{1}$ makes this not exactly true. Nonetheless, this probability term is still ``almost" the non-unique decoding error probability. We make this statement more precise. Let $W^n$,  $U^n$,  $V^n$, $\bar{W}^n$, and $\bar{U}^n$ denote respectively $W^n(1,1)$, $U^n(1,1,1)$, $V^n(1,1,1)$, $W^n(\bar{m},\bar{l}_0)$, and $U^n(\bar{m},\bar{l}_0,\bar{l}_1)$. The above probability term is upper-bounded by $$2^{n(R+T_0+T_1)}\sum_{(\bar{w}^n,\bar{u}^n,y_1^n)\in\mathcal{A}^n_\epsilon} p_{Y_1^n\bar{W}^n\bar{U}^n|\mathcal{I}}(y_1^n,\bar{w}^n,\bar{u}^n|\mathbf{1}),$$ and the inner pmf may be written as follows.
\begin{align*}
& p_{Y_1^n\bar{W}^n\bar{U}^n|\mathcal{I}}(y_1^n,\bar{w}^n,\bar{u}^n|\mathbf{1})\\&=  p_{Y_1^n|\mathcal{I}}(y_1^n|\mathbf{1})p_{\bar{W}^n|Y_1^n\mathcal{I}}(\bar{w}^n\!|y_1^n,\mathbf{1})  p_{\bar{U}^n|\bar{W}^nY_1^n\mathcal{I}}(\bar{u}^n\!|\bar{w}^n\!,y_1^n,\mathbf{1})\\&\stackrel{(a)}{\leq} p_{W^n}(\bar{w}^n)(1+c(\epsilon)) p_{Y_1^n|\mathcal{I}}(y_1^n|\mathbf{1})  p_{U^n|W^n}(\bar{u}^n|\bar{w}^n)
\end{align*}
In the above inequality, $(a)$ follows by  $p_{\bar{W}^n|Y_1^n\mathcal{I}}(\bar{w}^n|y_1^n,\mathbf{1})\leq (1+c(\epsilon))p_{{W}^n}(\bar{w}^n)$ (see \cite[Lemma 1]{MineroLimKim13}), and $c(\epsilon)\to0$ as $n$ grows large. Standard typicality arguments then bound this term (for every $\epsilon>0$ and for some large enough $n$) by $ (1+c(\epsilon))2^{n(R+T_0+T_1)}2^{-n(I(WU;Y)-\gamma_1(\epsilon))}.$

To analyze te second probability term, let $\tilde{W}^n$ and $\hat{U}^n$ denote $W^n(\tilde{m},\tilde{l}_0)$ and $U^n(1,1,\hat{l}_1)$ respectively. The second term of \eqref{term1new} is bounded from above by 
\begin{align*}
&2^{n(R+T_0+T_1)}\nonumber\\&\times\!\!\!\!\!\sum_{(\! w^n \!\!, s^n \!\!, u^n\!\! , v^n\!\! , y_1^n\! )\in\mathcal{A}^n_\epsilon}\!\!\sum_{\substack{\tilde{w}^n:\\(\! \tilde{w}^n \!\!, y_1^n \!)\in\mathcal{A}^n_\epsilon}}\!\sum_{\substack{\hat{u}^n:\\(\! w^n \!\!, \hat{u}^n\!\! , y_1^n\! )\in\mathcal{A}^n}} \!\!\!\!\!\!\!\!\! p(\!w^n\!\!,s^n\!\!,u^n\!\!,v^n\!\!,y_1^n\!,\tilde{w}^n\!\!,\hat{u}^n\!|\hspace{-0.025cm}\mathbf{1}\!),\end{align*} 
and we treat the inner pmf in a similar way as in Section \ref{onlinedecoder}.
\begin{align*}
&p_{W^nS^nU^nV^nY^n_1\tilde{W}^n\hat{U}^n|\mathcal{I}}(w^n,s^n,u^n,v^n,y_1^n,\tilde{w}^n,\hat{u}^n|\mathbf{1})\\
&=  p({w}^n,s^n,u^n,v^n,y_1^n|\mathbf{1})  p(\tilde{w}^n|w^n,s^n,u^n,v^n,y_1^n,\mathbf{1})\\&\quad\times  p(\hat{u}^n|\tilde{w}^n,w^n,s^n,u^n,v^n,y_1^n,\mathbf{1})
\end{align*}
It is now easy to see that (e.g., see \cite{MineroLimKim13}) \begin{eqnarray*} p(\tilde{w}^n|w^n,s^n,u^n,v^n,y_1^n,\mathbf{1})&=&p(\tilde{w}^n|w^n,s^n,\mathbf{1})\\&\leq&  (1+c(\epsilon))p_{\tilde{W}^n}(\tilde{w}^n).\end{eqnarray*}  Similarly, it turns out that (see Appendix \ref{app:typicalitylemma} and follow a similar line of argument)
\begin{eqnarray*} p(\hat{u}^n|\tilde{w}^n,w^n,s^n,u^n,v^n,y_1^n,\mathbf{1})&=&p(\hat{u}^n|w^n,s^n,u^n,v^n,\mathbf{1})\\&\leq&  2^{n\delta(\epsilon)} p_{\hat{U}^n|W^n}(\hat{u}^n|w^n).\end{eqnarray*} Therefore, the second term of \eqref{term1new} is bounded by $ 2^{n(R+T_0+T_1)}2^{-n(I(WU;Y)-\gamma_2(\epsilon)-\delta(\epsilon))}$.

One sees that the non-unique decoding constraints are sufficient to drive both terms of \eqref{term1new} to zero, as $n$ goes large.}

\section{The second probability term of inequality
\eqref{jrnl1-peexample3} can be made arbitrarily small by choosing
sufficiently large $n$ under the non-unique decoding constraints in
\cite{BandemerElGamal11}}
\label{jrnl1-appendixlast}
To upper-bound the second probability term of inequality
\eqref{jrnl1-peexample3}, we use union bound and inclusion of events to
obtain the expression in \eqref{jrnl1-peexample3term1}. We then show that each probability
term of inequality \eqref{jrnl1-peexample3term1} can be made arbitrarily
small by choosing a sufficiently large $n$, if the non-unique decoding constraints of
\cite{BandemerElGamal11} hold.
\allowdisplaybreaks
{{\begin{align}
&\Pr\!\! \left(\!\!\!\!\!\left.\begin{array}{c}(Q^n\!,X^n_1(\bar{m}_1),Y^n_1) \in \mathcal{A}^n_\epsilon\\\text{ for some }\bar{m}_1 \neq  1,\text{ and}\\ (Q^n\!,X^n_1(\breve{m}_1),X^n_{21}(\breve{m}_2),Y^n_1) \in \mathcal{A}^n_\epsilon\\\text{ for some }(\breve{m}_1,\breve{m}_2) \neq  (1,1),\text{ and}\\ (Q^n\!,X^n_1(\dot{m}_1),X^n_{31}(\dot{m}_3),Y^n_1) \in \mathcal{A}^n_\epsilon\\\text{ for some }(\dot{m}_1,\dot{m}_3) \neq  (1,1),\text{ and}\\(Q^n\!,X^n_1(\hat{m}_1),S^n_1(\hat{m}_2,\hat{m}_3)\hspace{1.8cm}\\\hspace{1cm},X^n_{21}(\hat{m}_2),X^n_{31}(\hat{m}_3),Y^n_1) \!\in\! \mathcal{A}^n_\epsilon\\\text{ for some }(\hat{m}_2,\hat{m}_3) \!\neq\!  (1,1),\,  \hat{m}_1 \!=\! 1\end{array}\!\right|\mathcal{I}=\mathbf{1}\!\!\right)\nonumber\\
&\leq\Pr\!\! \left(\!\!\!\!\!\left.\begin{array}{c}(Q^n\!,X^n_1(\bar{m}_1),Y^n_1) \in \mathcal{A}^n_\epsilon\\\text{ for some }\bar{m}_1 \neq  1,\text{ and}\\ (Q^n\!,X^n_1(\dot{m}_1),X^n_{31}(\dot{m}_3),Y^n_1) \in \mathcal{A}^n_\epsilon\\\text{ for some }(\dot{m}_1,\dot{m}_3) \neq  (1,1),\text{ and}\\(Q^n\!,X^n_1(\hat{m}_1),S^n_1(\hat{m}_2,\hat{m}_3)\hspace{1.8cm}\\\hspace{1cm},X^n_{21}(\hat{m}_2),X^n_{31}(\hat{m}_3),Y^n_1) \!\in\! \mathcal{A}^n_\epsilon\\\text{ for some }\hat{m}_2\! \neq\!  1,\, \hat{m}_3 \!=\! 1,\,  \hat{m}_1\! =\! 1\end{array}\!\!\!\right|\mathcal{I}=\mathbf{1}\!\!\right)\nonumber \\
&\quad+\Pr\!\! \left(\!\!\!\!\!\left.\begin{array}{c}(Q^n\!,X^n_1(\bar{m}_1),Y^n_1) \in \mathcal{A}^n_\epsilon\\\text{ for some }\bar{m}_1 \neq  1,\text{ and}\\ (Q^n\!,X^n_1(\breve{m}_1),X^n_{21}(\breve{m}_2),Y^n_1) \in \mathcal{A}^n_\epsilon\\\text{ for some }(\breve{m}_1,\breve{m}_2) \neq  (1,1),\text{ and}\\ (Q^n\!,X^n_1(\hat{m}_1),S^n_1(\hat{m}_2,\hat{m}_3)\hspace{1.8cm}\\\hspace{1cm},X^n_{21}(\hat{m}_2),X^n_{31}(\hat{m}_3),Y^n_1) \!\in\! \mathcal{A}^n_\epsilon\\\text{ for some }\hat{m}_2 \!=\! 1,\, \hat{m}_3 \!\neq\!  1,\,  \hat{m}_1 \!=\! 1\end{array}\!\!\!\right|\mathcal{I}=\mathbf{1}\!\!\right)\nonumber \\
&\quad+\Pr\!\! \left(\!\!\!\!\!\left.\begin{array}{c}(Q^n\!,X^n_1(\bar{m}_1),Y^n_1) \in \mathcal{A}^n_\epsilon\\\text{ for some }\bar{m}_1 \neq  1,\text{ and}\\(Q^n\!,X^n_1(\hat{m}_1),S^n_1(\hat{m}_2,\hat{m}_3)\hspace{1.8cm}\\\hspace{1cm},X^n_{21}(\hat{m}_2),X^n_{31}(\hat{m}_3),Y^n_1) \!\in\! \mathcal{A}^n_\epsilon\\\text{ for some }\hat{m}_2 \!\neq\!  1,\, \hat{m}_3 \!\neq\!  1,\,  \hat{m}_1 \!=\! 1\end{array}\!\!\!\right|\!\mathcal{I}=\mathbf{1}\!\!\right)
\label{jrnl1-peexample3term1}
\end{align}}}

The first  probability term of
 \eqref{jrnl1-peexample3term1} (and similary the second term) is analyzed below. 
 \begin{align}
&\Pr \left(\!\!\!\left.\begin{array}{c}(Q^n\!,X^n_1(\bar{m}_1),Y^n_1)\in\mathcal{A}^n_\epsilon\\\text{ for some }\bar{m}_1 \neq  1,\text{ and}\\ (Q^n\!,X^n_1(\dot{m}_1),X^n_{31}(\dot{m}_3),Y^n_1)\in\mathcal{A}^n_\epsilon\\\text{ for some }(\dot{m}_1,\dot{m}_3) \neq  (1,1),\text{ and}\\(Q^n\!,X^n_1(\hat{m}_1),S^n_1(\hat{m}_2,\hat{m}_3)\hspace{1.8cm}\\\hspace{1cm},X^n_{21}(\hat{m}_2),X^n_{31}(\hat{m}_3),Y^n_1)\!\in\!\mathcal{A}^n_\epsilon\\\text{ for some }\hat{m}_2 \!\neq\!  1,\ \hat{m}_3 \!=\! 1,\  \hat{m}_1\! = \!1\end{array}\!\!\!\right|\!\mathcal{I}\!=\!\mathbf{1}\right)\label{jrnl1-derivationexample21}\\
&\leq \Pr\! \left(\!\!\!\left.\begin{array}{c} (Q^n\!,X^n_1(\dot{m}_1),X^n_{31}(\dot{m}_3),Y^n_1)\!\in\!\mathcal{A}^n_\epsilon\\\text{ for some }\dot{m}_1 \neq  1,\ \dot{m}_3 = 1,\text{ and}\\(Q^n\!,X^n_1(\hat{m}_1),S^n_1(\hat{m}_2,\hat{m}_3)\hspace{1.8cm}\\\hspace{1cm},X^n_{21}(\hat{m}_2),X^n_{31}(\hat{m}_3), Y^n_1)\!\in\!\mathcal{A}^n_\epsilon\\\text{ for some }\hat{m}_2 \!\neq\!  1,\, \hat{m}_3 \!= \!1,\,  \hat{m}_1\! =\! 1\end{array}\hspace{0cm}\!\!\right|\!\mathcal{I}\!=\!\mathbf{1}\right)\label{57}\\
& +\Pr\! \left(\!\!\!\left.\begin{array}{c} ( Q^n\!,X^n_1(\dot{m}_1),X^n_{31}(\dot{m}_3),Y^n_1 )\!\in\!\mathcal{A}^n_\epsilon\\\text{ for some }\ \dot{m}_1 \neq  1,\ \dot{m}_3 \neq 1,\text{ and}\\( Q^n\!,X^n_1(\hat{m}_1),S^n_1(\hat{m}_2,\hat{m}_3)\hspace{1.8cm}\\\hspace{1cm},X^n_{21}(\hat{m}_2),X^n_{31}(\hat{m}_3), Y^n_1 )\!\in\!\mathcal{A}^n_\epsilon\\\text{ for some }\hat{m}_2 \!\neq  \!1,\, \hat{m}_3 \!=\! 1,\, \hat{m}_1 \!=\! 1\end{array}\!\!\!\right|\!\mathcal{I}\!=\!\mathbf{1}\right)\label{58}\\
& +\Pr\! \left(\!\!\!\left.\begin{array}{c}( Q^n\!,X^n_1(\bar{m}_1),Y^n_1)\!\in\!\mathcal{A}^n_\epsilon\\\text{ for some }\bar{m}_1 \neq  1,\text{ and}\\ ( Q^n\!,X^n_1(\dot{m}_1),X^n_{31}(\dot{m}_3),Y^n_1 )\!\in\!\mathcal{A}^n_\epsilon\\\text{ for some }\ \dot{m}_1 = 1,\ \dot{m}_3 \neq 1,\text{ and}\\( Q^n\!,X^n_1(\hat{m}_1),S^n_1(\hat{m}_2,\hat{m}_3)\hspace{1.8cm}\\\hspace{1cm},X^n_{21}(\hat{m}_2),X^n_{31}(\hat{m}_3), Y^n_1 )\!\in\!\mathcal{A}^n_\epsilon\\\text{ for some }\hat{m}_2 \!\neq \! 1,\, \hat{m}_3\! =\! 1,\,  \hat{m}_1 \!=\! 1\end{array}\!\!\!\right|\!\mathcal{I}\!=\!\mathbf{1}\right)\label{59}\\
&\leq\!\! 2^{nR_1}2^{n\min\{R_2,H(X_{21}|Q)\}}2^{-nI(X_1X_{21};Y_1|QX_{31})+2n\delta(\epsilon)}\label{60}\\
&+2^{nR_1}2^{n\min\{R_3,H(X_{31}|Q)\}}2^{n\min\{R_2,H(X_{21}|Q)\}}\nonumber\\&\hspace{2cm}\times2^{-nI(X_1X_{21}X_{31};Y_1|Q)+2n\delta(\epsilon)}\label{61}\\
&+2^{nR_1}2^{n\min\{R_3,H(X_{31}|Q)\}}2^{n\min\{R_2,H(X_{21}|Q),H(S_1|X_{31}Q)\}}\nonumber\\&\hspace{2cm}\times2^{-nI(X_1X_{21}X_{31};Y_1|Q)+2n\delta(\epsilon)}\label{jrnl1-derivationexample2last}
\end{align}
where $\delta(\epsilon)$ vanishes to zero as $\epsilon\to 0$.
The first inequality above is obtained by considering the different cases of $(\dot{m}_1,\dot{m}_2)$, and using inclusion of events. In the second inequality, the probability terms in \eqref{57}, \eqref{58},  and \eqref{59} are bounded by \eqref{60}, \eqref{61}, and \eqref{jrnl1-derivationexample2last}, respectively. Essentially, the derivation follows from an analysis similar to that of
\eqref{jrnl1-pe-term2} in Section \ref{onlinedecoder}, together with the bounding techniques of \cite{BandemerElGamal11} (where the key is in that depending on the input pmfs and the message rates, the number of possible combined interference sequences can be equal to the number of interfering message pairs, the number of typical combined interference sequences, or some combination of the two-- see \cite[Lemma 2 and Lemma~3]{BandemerElGamal11}).
 Here, we briefly outline how \eqref{57} is bounded by \eqref{60}, and we leave the derivation of the other two terms to the interested reader. 

 We start with the following bound.
\begin{align}
&\Pr\left(\hspace{-.2cm}\left.\begin{array}{c} ( Q^n\!,X^n_1(\dot{m}_1),X^n_{31}(\dot{m}_3), Y^n_1 )\!\in\!\mathcal{A}^n_\epsilon\\\text{ for some }\dot{m}_1 \neq  1,\ \dot{m}_3 = 1,\text{ and}\\( Q^n\!,X^n_1(\hat{m}_1),S^n_1(\hat{m}_2,\hat{m}_3)\hspace{1cm}\\\hspace{1cm},X^n_{21}(\hat{m}_2),X^n_{31}(\hat{m}_3), Y^n_1 )\!\in\!\mathcal{A}^n_\epsilon\\\text{ for some }\hat{m}_2 \!\neq\!  1,\, \hat{m}_3 \!=\! 1,\,  \hat{m}_1 \!=\! 1\end{array}\!\!\!\right| \mathcal{I}=\mathbf{1}\right)\nonumber\\
&\leq2^{nR_1}\!\Pr\!\! \left(\left.\hspace{-.35cm}\begin{array}{c} (Q^n,\dot{X}^n_1,X^n_{31},Y^n_1)\!\in\!\mathcal{A}^n_\epsilon,\text{ and}\\(Q^n,X^n_1,X^n_{21}(\hat{m}_2),X^n_{31},Y^n_1)\!\in\!\mathcal{A}^n_\epsilon\\\text{ for some }\hat{m}_2 \neq  1\end{array}\!\!\!\right|\mathcal{I}\!=\!\mathbf{1}\hspace{-.1cm}\right)\label{tired0}
\end{align}
Using the bounding technique of \cite{BandemerElGamal11}, the above probability term can be upper bounded in two different manners. By counting the number of different messages $\hat{m}_2$, we find\begin{align}
&\Pr \left(\left.\hspace{-.2cm}\begin{array}{c} (Q^n,\dot{X}^n_1,X^n_{31},Y^n_1)\!\in\!\mathcal{A}^n_\epsilon,\text{ and}\\(Q^n,X^n_1,X^n_{21}(\hat{m}_2),X^n_{31},Y^n_1)\!\in\!\mathcal{A}^n_\epsilon\\\text{ for some }\hat{m}_2 \neq  1\end{array}\!\!\right|\mathcal{I}=\mathbf{1}\hspace{0cm}\right)\nonumber\\
&\leq 2^{nR_2}\Pr \left(\left.\hspace{-.2cm}\begin{array}{c} (Q^n,\dot{X}^n_1,X^n_{31},Y^n_1)\!\in\!\mathcal{A}^n_\epsilon,\text{ and}\\(Q^n,X^n_1,\hat{X}^n_{21},X^n_{31},Y^n_1)\!\in\!\mathcal{A}^n_\epsilon\end{array}\right|\mathcal{I}=\mathbf{1}\hspace{-0cm}\right)\nonumber\\
&\leq 2^{nR_2}2^{-nI(X_1X_{21};Y_1|QX_{31})+2n\delta(\epsilon)}.\label{tired1}
\end{align}
Furthermore, by counting the number of typical sequences $X_{21}^n$, we find
\begin{align}
&\Pr \left(\left.\hspace{-.2cm}\begin{array}{c} (Q^n,\dot{X}^n_1,X^n_{31},Y^n_1)\!\in\!\mathcal{A}^n_\epsilon,\text{ and}\\(Q^n,X^n_1,X^n_{21}(\hat{m}_2),X^n_{31},Y^n_1)\!\in\!\mathcal{A}^n_\epsilon\\\text{ for some }\hat{m}_2 \neq  1\end{array}\!\!\right|\mathcal{I}=\mathbf{1}\hspace{0cm}\right)\nonumber\\
&\leq \!\!\sum_{(q^n,x_{31}^n)\in\mathcal{A}_\epsilon^n}\!\!\left[p(q^n,x^n_{31})\vphantom{\times\Pr \left(\left.\hspace{-.2cm}\begin{array}{c} (q^n,\dot{X}^n_1,x^n_{31},Y^n_1)\!\in\!\mathcal{A}^n_\epsilon\\(q^n,X^n_1,\hat{X}^n_{21}(\hat{m}_2)\hspace{.8cm}\\\hspace{1cm},x^n_{31},Y^n_1)\!\in\!\mathcal{A}^n_\epsilon\\\text{ for some }\hat{m}_2\neq 1\end{array}\!\!\!\!\right|\!\begin{array}{l}Q^n\!=\!q^n,\\ X_3^n\!=\!x_3^n,\\\mathcal{I}\!=\!\mathbf{1}\end{array}\hspace{-.2cm}\right)}\right.\nonumber\\&\hspace{1.9cm}\left.\times\Pr \left(\left.\hspace{-.2cm}\begin{array}{c} (q^n,\dot{X}^n_1,x^n_{31},Y^n_1)\!\in\!\mathcal{A}^n_\epsilon\\(q^n,X^n_1,\hat{X}^n_{21}(\hat{m}_2)\hspace{.8cm}\\\hspace{1cm},x^n_{31},Y^n_1)\!\in\!\mathcal{A}^n_\epsilon\\\text{ for some }\hat{m}_2\neq 1\end{array}\!\!\!\!\right|\!\begin{array}{l}Q^n\!=\!q^n,\\ X_3^n\!=\!x_3^n,\\\mathcal{I}\!=\!\mathbf{1}\end{array}\hspace{-.2cm}\right)\!\!\right]\nonumber\\
&\leq2^{nH(X_{21}|Q)}2^{-nI(X_1X_{21};Y_1|QX_{31})+2n\delta(\epsilon)}.\label{tired2}
\end{align}
Putting together \eqref{tired0}, \eqref{tired1}, and \eqref{tired2} results in the bound \eqref{60}.

Finally, the third probability term of  \eqref{jrnl1-peexample3term1} is
bounded as follows.
\begin{align}
& \Pr\!\!\left(\!\left.\hspace{-.3cm}\begin{array}{c}( Q^n\!,X^n_1(\bar{m}_1), Y^n_1 )\in\mathcal{A}^n_\epsilon\\\text{ for some }\bar{m}_1\neq ,\text{ and}\\(\!Q^n\! , X^n_1 ( \hat{m}_1 ) , S^n_1 ( \hat{m}_2 , \hat{m}_3 ) \hspace{1.8cm}\\\hspace{1cm}, X^n_{21} ( \hat{m}_2 ) , X^n_{31} ( \hat{m}_3 ) , Y^n_1\!)\!\in\!\mathcal{A}^n_\epsilon\\\text{ for some }\hat{m}_2\!\neq\! 1,\, \hat{m}_3\!\neq\! 1,\,  \hat{m}_1\!=\!1\end{array}\hspace{-.2cm}\right| \mathcal{I}\!=\!\mathbf{1}\hspace{-.15cm}\right) \label{jrnl1-onetolastderivation}\\
&\leq\Pr\!\left(\!\!\left.\hspace{-.2cm}\begin{array}{c}(Q^n,X^n_1(\bar{m}_1),Y^n_1)\in\mathcal{A}^n_\epsilon\\\text{ for some }\tilde{m}_1\neq 1,\text{ and}\\(Q^n,X^n_1(\hat{m}_1),S^n_1(\hat{m}_2,\hat{m}_3),Y^n_1)\in\mathcal{A}^n_\epsilon\\\text{ for some }\hat{m}_2\neq 1,\ \hat{m}_3\neq 1,\  \hat{m}_1=1\end{array}\right|\mathcal{I}=\mathbf{1}\hspace{0cm}\right) \\
&\leq 2^{nR_1}2^{n\min\{R_2+R_3,R_2+H(X_{31}|Q),H(X_{21}|Q)+R_3,H(S_1|Q)\}}\nonumber\\&\quad\times2^{-nI(X_1S_1;Y_1|Q)+\delta(\epsilon)}\label{jrnl1-lastderivation},
\end{align}
where $\delta(\epsilon)\to 0$ when $\epsilon\to 0$.

\section{{The second probability term in \eqref{forlastapp} is bounded by $2^{n(R_1+T_1+R_r-I(V_0V_1;Y_1|U)+\gamma_2(\epsilon)+\delta(\epsilon))}$}}
\label{app-last-example3}
We now show that for any $\epsilon>0$, the second probability term in \eqref{forlastapp} is bounded (for a large enough $n$) by $2^{n(R_1+T_1+R_r-I(V_0V_1;Y_1|U)+\gamma_2(\epsilon)+\delta(\epsilon))}$. Let us denote $U^n(1)$, $V_0^n(1,1,1)$, $V_0^n(1,\tilde{m}_1,\tilde{m}_r)$, and $V_1^n(1,1,1,\hat{t}_1)$  by $U^n$, $V_0^n$, $\tilde{V}_0^n$, and $\hat{V}^n_1$ respectively. Proceeding as in Section \ref{onlinedecoder}, we bound \eqref{forlastapp} in \eqref{firstderiv}-\eqref{lastderiv} at the top of Page \pageref{lastderiv}. Note that inequality \eqref{tired3} follows for the same reasons as in Appendix~\ref{app:typicalitylemma}.

\section{The probability that the joint unique decoder of $m_0$, $m_1$ and $m_r$ in Subsection \ref{jrnl1-example1} fails}
\label{jrnl1-appendixexample1etc}
We analyze the probability that a joint unique decoder fails to uniquely decode indices $m_0$, $m_1$, $m_r$ and show that it fails with high probability if either \eqref{jrnl1-r1} or \eqref{jrnl1-r2} is violated. 
Note that
\begin{align}
&\hspace{-.2cm}\Pr\!\left(\!\!\!\!\!   \left.\begin{array}{l}(U^n(\hat{m}_0),V^n_0(\hat{m}_0,\hat{m}_1,\hat{m}_r),Y^n_1)\!\in\! \mathcal{A}^n_\epsilon\\\text{ for some } (\hat{m}_0,\hat{m}_1,\hat{m}_r) \!\neq\!  (1,1,1)\end{array}\!\! \!\right|\mathcal{I}=\mathbf{1}\!\!  \right)\nonumber\\&\geq\Pr\!\left(\!\!\!\!\!\!\left.   \begin{array}{l}(U^n(\hat{m}_0),V^n_0(\hat{m}_0,\hat{m}_1,\hat{m}_r),Y^n_1)\!\in\! \mathcal{A}^n_\epsilon\\\text{ for some } (\hat{m}_0,\hat{m}_1,\hat{m}_r)\! \neq\!  (1,1,1),\, \hat{m}_0 \!=\! 1\end{array}\!\!\!\right|\mathcal{I}=\mathbf{1}\!\!   \right),\label{jrnl1-term1}
\end{align}
and
\begin{align}
&\hspace{-.2cm}\Pr\!\left(\!\!\!\!\!  \left.\begin{array}{l}(U^n(\hat{m}_0),V^n_0(\hat{m}_0,\hat{m}_1,\hat{m}_r),Y^n_1)\!\in\! \mathcal{A}^n_\epsilon\\\text{ for some } (\hat{m}_0,\hat{m}_1,\hat{m}_r) \!\neq\!  (1,1,1)\end{array}\!\!\!  \right|\mathcal{I}=\mathbf{1} \!\!\right)\nonumber\\&\geq\Pr\!\left(\!\!\!\!\!\!\left.   \begin{array}{l}(U^n(\hat{m}_0),V^n_0(\hat{m}_0,\hat{m}_1,\hat{m}_r),Y^n_1)\!\in\! \mathcal{A}^n_\epsilon\\\text{ for some } (\hat{m}_0,\hat{m}_1,\hat{m}_r) \!\neq\!  (1,1,1),\, \hat{m}_0 \!\neq\! 1\end{array} \!\!\!\right|\mathcal{I}=\mathbf{1}\!\!  \right).\label{jrnl1-term2}
\end{align}
It is now not hard to see that the probability term on the right hand side of inequality \eqref{jrnl1-term1} is arbitrarily close to $1$ if $R_1+R_r>I(V_0;Y_1|U)$ and  the probability term on the right hand side of inequality \eqref{jrnl1-term2} is arbitrarily close to $1$ if $R_1+R_r>I(UV_0;Y_1)$.

\begin{widetext}
\begin{align}
&\Pr\!\left( \left.\begin{array}{c} (U^n(1),V_0^n(1,1,1),V_1^n(1,1,1,1),V_2^n(1,1,1,1),Y^n_1)\in\mathcal{A}_\epsilon^n\text{ and}\\(U^n(\tilde{m}_0),V^n_0(\tilde{m}_0,\tilde{m}_1,\tilde{m}_r),Y^n_1)\!\in\! A^n_\epsilon \\ {\text{for some }(\tilde{m}_0,\tilde{m}_1,\tilde{m}_r)\!\neq\!(1,1,1),\ \tilde{m_0}\!=\!1}\text{ and} \\(U^n(\hat{m}_0),V^n_0(\hat{m}_0,\hat{m}_1,\hat{m}_r),V^n_1(\hat{m}_0,\hat{m}_1,\hat{m}_r,\hat{t}_1),Y^n_1)\!\in\! A^n_\epsilon \\ {\text{for some }(\hat{m}_0,\hat{m}_1,\hat{m}_r)\!=\!(1,1,1),\ \!\hat{t}_1\!\neq\! 1}\end{array}\right|\mathcal{I}=\mathbf{1} \right)\label{firstderiv}\\
&\leq2^{n(R_1+R_r+T_1)}\Pr\left(\left.\begin{array}{c}(U^n(1),V_0^n(1,1,1),V_1^n(1,1,1,1),V_2^n(1,1,1,1),Y^n_1)\in\mathcal{A}_\epsilon^n\text{ and}\\(U^n(1),V^n_0(1,\tilde{m}_1,\tilde{m}_r),Y^n_1)\!\in\! A^n_\epsilon\text{ and }\\ (U^n(1),V^n_0(1,1,1),V^n_1(1,1,1,\hat{t}_1),Y^n_1)\!\in\! A^n_\epsilon\end{array}\right|\mathcal{I}=\mathbf{1}\right)\\
&\leq2^{n(R_1+R_r+T_1)}\hspace{-.75cm}\sum_{(u^n\!,v_0^n,v_1^n,v_2^n,y_1^n)\in\mathcal{A}^n_\epsilon}\sum_{\substack{\tilde{v}_0^n:\\(u^n\!,\tilde{v}_0^n,y_1^n)\in\mathcal{A}^n_\epsilon}}\sum_{\substack{\hat{v}_1^n:\\(u^n\!,v_0^n,\hat{v}_1^n,y_1^n)\in\mathcal{A}^n_\epsilon}}\hspace{-.75cm}p_{U^nV_0^nV_1^nV_2^nY_1^n\tilde{V}^n_0\hat{V}_1^n|\mathcal{I}}(u^n\!,v_0^n,v_1^nv_2^n,y^n_1,\tilde{v}_0^n,\hat{v}_1^n|\mathbf{1})\\
&\leq2^{n(R_1+R_r+T_1)}\hspace{-.75cm}\sum_{(u^n\!,v_0^n,v_1^n,v_2^n,y_1^n)\in\mathcal{A}^n_\epsilon}\sum_{\substack{\tilde{v}_0^n:\\(u^n\!,\tilde{v}_0^n,y_1^n)\in\mathcal{A}^n_\epsilon}}\sum_{\substack{\hat{v}_1^n:\\(u^n\!,v_0^n,\hat{v}_1^n,y_1^n)\in\mathcal{A}^n_\epsilon}}\hspace{-.75cm}\substack{p(u^n\!,v_0^n,v_1^n,v_2^n,y^n_1|\mathbf{1})p(\tilde{v}_0^n|u^n\!,v_0^n,v_1^n,v_2^n,y^n_1,\mathbf{1})\\\times p(\hat{v}_1^n|u^n\!,v_0^n,v_1^n,v_2^n,y^n_1,\tilde{v}_0^n,\mathbf{1})}\\
&=2^{n(R_1+R_r+T_1)}\hspace{-.75cm}\sum_{(u^n\!,v_0^n,v_1^n,v_2^n,y_1^n)\in\mathcal{A}^n_\epsilon}\sum_{\substack{\tilde{v}_0^n:\\(u^n\!,\tilde{v}_0^n,y_1^n)\in\mathcal{A}^n_\epsilon}}\sum_{\substack{\hat{v}_1^n:\\(u^n\!,v_0^n,\hat{v}_1^n,y_1^n)\in\mathcal{A}^n_\epsilon}}\hspace{-.75cm}p(u^n\!,v_0^n,v_1^n,v_2^n,y^n_1|\mathbf{1})p(\tilde{v}_0^n|u^n)p(\hat{v}_1^n|u^n\!,v_0^n,v_1^n,v_2^n,\mathbf{1})\\
&\leq2^{n(R_1+R_r+T_1)}\hspace{-.75cm}\sum_{(u^n\!,v_0^n,v_1^n,v_2^n,y_1^n)\in\mathcal{A}^n_\epsilon}\sum_{\substack{\tilde{v}_0^n:\\(u^n\!,\tilde{v}_0^n,y_1^n)\in\mathcal{A}^n_\epsilon}}\label{tired3}\sum_{\substack{\hat{v}_1^n:\\(u^n\!,v_0^n,\hat{v}_1^n,y_1^n)\in\mathcal{A}^n_\epsilon}}\hspace{-.75cm}2^{n\delta(\epsilon)}p(u^n\!,v_0^n,v_1^n,v_2^n,y^n_1|\mathbf{1})p(\tilde{v}_0^n|u^n)p(\hat{v}_1^n|u^n\!,v_0^n)\\
&\leq2^{n(R_1+R_r+T_1)}2^{-n(I(V_0V_1;Y|U)-\gamma_2(\epsilon)-\delta(\epsilon))}\label{lastderiv}
\end{align}
\end{widetext}

\section{Chernoff Bounds and inequalities \eqref{lastineq1} and \eqref{lastineq2}}
\label{LASTAPP}
Let $N=2^{n(T_3-S_3)}$, $M=2^{n(T_2-S_2)}$.
To simplify notation, we define $X_{i,j}$ to be a binary random variable which takes value $0$ when $(V_2^n(1,1,i),V_3^n(1,1,j))\in\mathcal{A}_\epsilon^n$. For example, $X_{1,1}=1$ by the assumption that $(v_2^n,v_3^n)\in\mathcal{A}_\epsilon^n$.  Also, $N_2(\hat{V}_2^n,v_2^n,v_3^n,\mathbf{C}^\prime)=\sum_{t_3=1}^NX_{2,t_3}$ and $N_3(v_3^n,\mathbf{C}^\prime)=\sum_{t_3=1}^NX_{3,t_3}$. Furthermore, we define $p_{\tilde{v}_2^n}=\Pr((\tilde{v}_2^n,V_3^n(1,1,2))\in\mathcal{A}_\epsilon^n|V_2^n(1,1,2)=\tilde{v}_2^n,U^n=u^n)$. For $\epsilon_1$-typical sequences $\tilde{v}_2^n$ (where $\epsilon_1<\epsilon$), we have $2^{-n(I(V_2;V_3|U)+\delta(\epsilon))}\leq p_{\tilde{v}_2^n}\leq 2^{-n(I(V_2;V_3|U)-\delta(\epsilon))}$. We let $p_l=2^{-n(I(V_2;V_3|U)+\delta(\epsilon))}$ and $p_u=2^{-n(I(V_2;V_3|U)-\delta(\epsilon))}$.
To prove Claim \ref{dblexp}, we show that
\begin{align}
&\Pr\left(\!\!\!\!\begin{array}{l|l}N_2( \hat{V}_2^n, v_2^n, v_3^n, \mathbf{C}^\prime )\!>\!2Np_u\!\!&\!\!U^n\!=\!u^n\end{array}\!\!\!\!\right)\!\leq\! \beta_1\exp\!\left(-\alpha_1\! Np_l\right)\label{chernoff1}
\end{align}
for some $\alpha_1,\beta_1>0$, and
\begin{align}
&\Pr\!\left(\!\!\!\!\begin{array}{l|l}N_3(v_3^n,\mathbf{C}^\prime)\!<\!\frac{1}{2}Np_l\!\!&\!\!U^n\!=\!u^n\end{array}\!\!\!\!\right)\leq \beta_2\exp\left(-\alpha_2 Np_l\right)\label{chernoff2}
\end{align}
for some $\alpha_2,\beta_2>0$.

We start with \eqref{chernoff1}.
\allowdisplaybreaks
\begin{eqnarray*}
&&\Pr\left(N_2(\hat{V}_2^n,v_2^n,v_3^n,\mathbf{C}^\prime)>2Np_u\,\right|\left.\vphantom{N_2(\hat{V}_2^n,v_2^n,v_3^n,\mathbf{C}^\prime)>2Np_u}U^n=u^n\right)\\
&&=\Pr\left(\sum_{t_3=1}^NX_{2,t_3}>2Np_u\,\right|\left.\vphantom{\sum_{t_3=1}^NX_{2,t_3}>2Np_u}U^n=u^n\right)\\
&&\leq \Pr\left(\sum_{t_3= 2}^NX_{2,t_3}>2Np_u-1\right|\left.\vphantom{\sum_{t_3= 2}^NX_{2,t_3}>2Np_u-1}U^n=u^n\right)\\
&&\leq\frac{\mathbb{E}\left[e^{t\sum_{t_3=2}^NX_{2,t_3}}\right|\left.\vphantom{e^{t\sum_{t_3=2}^NX_{2,t_3}}}U^n=u^n\right]}{e^{t2Np_u-t}},\ t>0\\
&&=\frac{\mathbb{E}\left[\mathbb{E}\left[e^{t\sum_{t_3=2}^NX_{2,t_3}}\right|\left.\vphantom{e^{t\sum_{t_3=2}^NX_{2,t_3}}}\hat{V}_2^n,U^n=u^n\right]\right|\left.\vphantom{\mathbb{E}[e^{t\sum_{t_3=2}^NX_{2,t_3}}|\hat{V}_2^n,U^n=u^n]}U^n=u^n\right]}{e^{t2Np_u-t}}\\
&&=\frac{\mathbb{E}\left[\prod_{t_3=2}^N\mathbb{E}\left[e^{tX_{2,t_3}}\right|\left.\vphantom{e^{tX_{2,t_3}}}\hat{V}_2^n,U^n=u^n\right]\right|\left.\vphantom{\prod_{t_3=2}^N\mathbb{E}\left[e^{tX_{2,t_3}}|\hat{V}_2^n,U^n=u^n\right]}U^n=u^n\right]}{e^{t2Np_u-t}}\\
&&=\frac{\mathbb{E}\left[\prod_{t_3=2}^N\left(p_{\hat{V}_2^n}e^{t}+(1-p_{\hat{V}_2^n})\right)\right|\left.\vphantom{\prod_{t_3=2}^N\left(p_{\hat{V}_2^n}e^{t}+(1-p_{\hat{V}_2^n})\right)}U^n=u^n\right]}{e^{t2Np_u-t}}\\
&&\leq\frac{\left(1+p_{u}(e^{t}-1)\right)^N}{e^{t2Np_u-t}}
\end{eqnarray*}
Set $t=\frac{1}{2}$. Then
\begin{eqnarray*}
&&\Pr\left(N_2(\hat{V}_2^n,v_2^n,v_3^n,\mathbf{C}^\prime)>2Np_u\right|\left.\vphantom{N_2(\hat{V}_2^n,v_2^n,v_3^n,\mathbf{C}^\prime)>2Np_u}U^n=u^n\right)\\
&&\leq\frac{\left(1+p_{u}(e^{\frac{1}{2}}-1)\right)^N}{e^{-\frac{1}{2}+Np_u}}\\
&&= e^{\frac{1}{2}}\left(\frac{1+p_{u}(e^{\frac{1}{2}}-1)}{e^{p_u}}\right)^N\\
&&\leq e^{\frac{1}{2}}{e^{-Np_{u}(2-e^{\frac{1}{2}})}}\\
&&\leq \beta_1 e^{-\alpha_1 Np_{l}},\text { for  }\alpha_1=2-e^{\frac{1}{2}},\ \beta_1=e^{\frac{1}{2}}.\\
\end{eqnarray*}
Similarly, to show \eqref{chernoff2} we proceed as follows.
\allowdisplaybreaks
\begin{eqnarray*}
&&\Pr\left(N_3(v_3^n,\mathbf{C}^\prime)<\frac{1}{2}Np_l\,\right|\left.\vphantom{N_3(v_3^n,\mathbf{C}^\prime)<\frac{1}{2}Np_l}U^n=u^n\right)\\
&&=\Pr\left(\sum_{t_3=1}^NX_{3,t_3}<\frac{1}{2}Np_l\,\right|\left.\vphantom{\sum_{t_3=1}^NX_{3,t_3}<\frac{1}{2}Np_l}U^n=u^n\right)\\
&&\leq\Pr\left(\sum_{t_3= 2}^NX_{3,t_3}<\frac{1}{2}Np_l\,\right|\left.\vphantom{\sum_{t_3= 2}^NX_{3,t_3}<\frac{1}{2}Np_l}U^n=u^n\right)\\
&&\leq\frac{\mathbb{E}\left[e^{-t\sum_{t_3=2}^NX_{3,t_3}}\right|\left.\vphantom{e^{-t\sum_{t_3=2}^NX_{3,t_3}}}U^n=u^n\right]}{e^{-t\frac{1}{2}Np_l}},\quad t>0\\
&&=\frac{\mathbb{E}\!\left[\mathbb{E}\!\left[e^{-t\sum_{t_3= 2}^NX_{3,t_3}}\right|\left.\!\vphantom{e^{-t\sum_{t_3= 2}^NX_{3,t_3}}}{V}_2^n(1,1,3),U^n\!=\!u^n\right]\!\right|\left.\!\vphantom{\mathbb{E}\!\left[e^{-t\sum_{t_3= 2}^NX_{3,t_3}}\right|\left.\vphantom{e^{-t\sum_{t_3= 2}^NX_{3,t_3}}}{V}_2^n(1,1,3),U^n\!=\!u^n\right]\!\!}U^n\!=\!u^n\right]}{e^{-t\frac{1}{2}Np_l}}\\
&&=\frac{\mathbb{E}\!\left[\prod_{t_3= 2}^N\mathbb{E}\!\left[e^{-tX_{3,t_3}}\right|\left.\!\vphantom{e^{-tX_{3,t_3}}}{V}_2^n(1,1,3),U^n\!=\!u^n\right]\!\right|\left.\vphantom{\prod_{t_3= 2}^N\mathbb{E}\!\left[e^{-tX_{3,t_3}}\right|\left.\vphantom{e^{-tX_{3,t_3}}}{V}_2^n(1,1,3),U^n\!=\!u^n\right]\!}\!U^n\!=\!u^n\right]}{e^{-t\frac{1}{2}Np_l}}\\
&&=\frac{\mathbb{E}\left[\prod_{t_3=2}^N\left(1-p_{V_2^n(1,1,3)}(1-e^{-t})\right)\right|\left.\vphantom{\prod_{t_3=2}^N\left(1-p_{V_2^n(1,1,3)}(1-e^{-t})\right)}U^n=u^n\right]}{e^{-t\frac{1}{2}Np_l}}\\
&&\leq\frac{\left(1-p_{l}(1-e^{-t})\right)^N}{e^{-t\frac{1}{2}Np_l}}\\
\end{eqnarray*}
Set $t=1$. Then
\allowdisplaybreaks
\begin{eqnarray*}
&&\Pr\left(N_3(v_3^n,\mathbf{C}^\prime)<\frac{1}{2}Np_l\,\right|\left.\vphantom{N_3(v_3^n,\mathbf{C}^\prime)<\frac{1}{2}Np_l}U^n=u^n\right)\\
&&\leq\left(\frac{1-p_{l}(1-e^{-1})}{e^{-\frac{1}{2}p_l}}\right)^N\\
&&\leq e^{-p_{l}(\frac{1}{2}-e^{-1})N}\\
&&\leq \beta_2e^{-\alpha_2N p_{l}},\quad \text{ for  }\alpha_2=\frac{1}{2}-e^{-1},\ \beta_2=1.
\end{eqnarray*}